\documentclass[a4paper,10pt]{article}
\usepackage{test,framed,caption,tikz,subcaption,booktabs}

\usepackage[authoryear,round]{natbib}
\usepackage{color}
\usepackage[colorlinks,citecolor=blue,urlcolor=magenta]{hyperref}
\usepackage{doi}

\usepackage[inline,shortlabels]{enumitem}
\setlist[enumerate,1]{label=(\roman*)}

\title{Asymptotic Linearity of Consumption Functions and Computational Efficiency}
\author{Qingyin Ma\thanks{International School of Economics and Management, Capital University of Economics and Business. Email: \href{mailto:qingyin.ma@cueb.edu.cn}{qingyin.ma@cueb.edu.cn}.} \and Alexis Akira Toda\thanks{Department of Economics, University of California San Diego. Email: \href{mailto:atoda@ucsd.edu}{atoda@ucsd.edu}.}}

\numberwithin{equation}{section}
\numberwithin{thm}{section}
\numberwithin{exmp}{section}

\newcommand{\aCRRA}{\operatorname{aCRRA}}
\newcommand{\RV}{\operatorname{RV}}
\newcommand{\aE}{\operatorname{aE}}
\newcommand{\cA}{\mathcal{A}}
\newcommand{\cC}{\mathcal{C}}

\newcommand{\cS}{\mathcal{S}}
\newcommand{\ZZ}{\mathsf{Z}}

\begin{document}

\maketitle

\begin{abstract}
We prove that the consumption functions in optimal savings problems are asymptotically linear if the marginal utility is regularly varying. We also analytically characterize the asymptotic marginal propensities to consume (MPCs) out of wealth. Our results are useful for obtaining good initial guesses when numerically computing consumption functions, and provide a theoretical justification for linearly extrapolating consumption functions outside the grid.

\medskip

{\bf Keywords:} computational efficiency, optimal savings problem, regular variation.

\medskip

{\bf JEL codes:} C63, C65, D15.
\end{abstract}

\section{Introduction}
The optimal savings problem---a dynamic optimization problem in which an agent chooses the optimal level of consumption and savings---is a fundamental building block of modern macroeconomics and contributes to a wide range of research fields ranging from asset pricing, life-cycle choice, fiscal policy, social security, to income and wealth inequality, among others.\footnote{See, for example, \cite{DeatonLaroque1992RES, DeatonLaroque1996JPE},  \cite{CagettiDenardi2006JPE}, \cite{DenardiFrenchJones2010JPE}, \cite{GunerKaygusuzVentura2012}, \cite{GuvenenSmith2014ECMA}, \cite{HeathcoteStoreslettenViolante2014AER}, \cite{BenhabibBisinZhu2015}, and the overview of \cite{Guvenen2011}. } The last several decades have witnessed substantial development in the theory of optimal savings. At the same time, existing studies find supporting evidence that the optimal consumption function---solution to the optimal savings problem---is asymptotically linear in wealth in various specialized settings.

In simple analytically solvable models that feature homothetic preferences and no income risk as in \cite{samuelson1969}, it is well known that the marginal propensity to consume (MPC) out of wealth is independent of the wealth level. In more complicated models, the asymptotic linearity of consumption functions has been numerically observed as in \citet[Figure~IV]{Zeldes1989QJE}, \citet[Figure~1]{huggett1993}, \citet[Figure~2]{krusell-smith1998}, and \citet[Figure~4]{Toda2019JME}. With hyperbolic absolute risk aversion (HARA) preferences and general income shocks, \cite{CarrollKimball1996} show that the consumption functions are concave, which implies that the MPCs converge, although they do not characterize the limit. More recently, \cite{MaToda2021JET} establish the asymptotic linearity of consumption functions and analytically characterize the asymptotic MPCs when the utility function is constant relative risk aversion (CRRA).

In spite of these interesting findings, the asymptotic properties of the optimal consumption function have hitherto received no general investigation. One cost of this status quo is that in various applications, the asymptotic behavior of agents' consumption as asset tends to infinity has substantial impact on  economic activities. For example, when studying wealth inequality, the saving performance of the rich, which is closely related to the asymptotic MPCs, is a driving force of the fat-tailed wealth distribution and its evolution \citep{FagerengHolmMollNatvikWP}. Without a systematic understanding of the asymptotic properties of consumption, researchers will have to provide their own analysis piecemeal in individual applications.

A second cost is concerned with numerical computation. When solving for the optimal consumption function numerically, it is common to evaluate the functions on a finite grid and interpolate or extrapolate off the grid points. Some extrapolation is usually necessary because even if the agent's asset is currently inside the grid, when the return on wealth is sufficiently high, the next period's asset may fall outside the grid with positive probability. Having a theory of optimal consumption at infinity is useful because it tells us how to properly set up the grid points and extrapolate functions outside the grid. 

In this paper, we systematically study the asymptotic behavior of the optimal consumption function in a highly general framework that contains a wide range of important settings as special cases, including the settings of some recent advancements in optimal savings \citep{MaStachurskiToda2020JET, MaToda2021JET}. Our main result is that, under the weak assumption that the marginal utility asymptotically performs like a power function as consumption increases (plus some other regularity assumptions),\footnote{This 
	specification includes commonly used utility functions such as CRRA or HARA as special cases.}
the consumption functions are asymptotically linear, or equivalently, the asymptotic MPCs converge to some constants.\footnote{Throughout the paper 
	we say that a consumption function $c(a)$ (where $a>0$ is financial wealth) is \emph{asymptotically linear} if the asymptotic \emph{average} propensity to consume $\bar{c}=\lim_{a\to\infty}c(a)/a$ exists. This condition is weaker than $\lim_{a\to\infty}\abs{c(a)-\bar{c}a-d}=0$ for some $\bar{c},d\in \R$, which may be a more common definition of asymptotic linearity. If the asymptotic MPC $\bar{c}=\lim_{a\to\infty}c'(a)$ exists, then l'H\^opital's rule implies $\lim_{a\to\infty}c(a)/a=\lim_{a\to\infty}c'(a)=\bar{c}$. Although not necessarily mathematically precise, due to the lack of better language we use ``constant asymptotic average propensity to consume'', ``constant asymptotic MPC'', and ``asymptotic linearity'' interchangeably.} 
Furthermore, we analytically characterize the asymptotic MPCs. 

Different from the existing literature, which typically focuses on special utility functions such as CRRA or HARA in relatively stylized settings, we only require that the marginal utility function asymptotically behaves like a power function, which is mathematically defined as \emph{regular variation}. Our results are significantly more general than the existing literature because regular variation is a parametric assumption only at infinity, and we do not impose any assumptions on the utility function on compact sets beyond the usual monotonicity and concavity.

Furthermore, based on the theory we develop, we systematically study computation methods. We focus on both computation speed and solution accuracy. As to the former, we apply our theory to construct proper initial guesses that facilitate efficient computation. The initial guess we propose relies on the asymptotic MPCs we derive and can be solved conveniently in applications. Numerical experiments show that policy iteration via the initial guess we propose is about $1.3$ to $1.8$ times faster than via the routine initial guess of consuming all current assets. As to the latter, we study in depth how to properly set up the grid points and extrapolate policy functions outside the grid when solving models numerically. This is realized by comparing the distances of MPCs at different asset levels from their theoretical asymptotes, as well as by exploring how truncating the grid space affects solution accuracy (measured by the error of the calculated consumption function relative to the true consumption function). 

The theory we develop provides a theoretical justification for linearly extrapolating policy functions outside the grid
when solving optimal savings problems numerically as in \cite{Gouin-BonenfantTodaParetoExtrapolation}. A closely related contribution is that our theory explains the ``approximate aggregation'' property in heterogeneous-agent general equilibrium models as in \cite{krusell-smith1998}. Approximate aggregation refers to the observation that, when solving heterogeneous-agent general equilibrium models, keeping track of just the first moment of the wealth distribution is nearly sufficient, despite the fact that the entire wealth distribution is a state variable. Because the market clearing condition involves aggregate savings, aggregation would be possible if saving is linear in wealth. Our results show that consumption (hence saving) is approximately linear in wealth, which explains the approximate aggregation property. 

The rest of this paper is structured as follows. The remaining of this section discusses related literature. Section~\ref{sec:main_res} formulates the optimal savings problem and establishes our main theoretical results. Asymptotic properties of the optimal consumption function are studied in general settings. Sufficient conditions for asymptotic linearity of the consumption function are discussed. Section~\ref{sec:comp} inspects the computation method in detail. By applying our theory, we propose useful initial guesses for efficient computation and discuss various details concerning solution accuracy. Section~\ref{sec:conclu} concludes. Main proofs are deferred to the appendices.

\subsection*{Related literature}

The existence of a solution to optimal savings problems has been studied by \cite{SchechtmanEscudero1977}, \cite{ChamberlainWilson2000}, and \cite{LiStachurski2014}. The recent work \cite{MaStachurskiToda2020JET} extend the Euler equation method of \cite{LiStachurski2014} and show the existence and uniqueness of a solution in a general setting with Markovian shocks, capital income risk, stochastic discounting, and potentially unbounded utility functions. \cite{MaToda2021JET} make further extension to \cite{MaStachurskiToda2020JET} by relaxing their assumptions on utility and idiosyncratic risks. Our paper is in the spirit of \cite{MaToda2021JET}.

Because optimal savings problems generally do not admit closed-form solutions, proving properties of the theoretical solution is often challenging. \cite{Rabault2002} studies under what conditions borrowing constraints bind. \cite{BenhabibBisinZhu2015} characterize the tail behavior of the wealth distribution under \iid capital and labor income shocks, which \cite{MaStachurskiToda2020JET} extend to a Markovian setting. \cite{Holm2018} shows that with HARA preferences, tightening the liquidity constraint decreases consumption. \cite{Light2018} shows that when the marginal utility is convex and the Markov chain has a certain monotonicity property, increasing income risk increases precautionary savings. \cite{LehrerLight2018} show that with CRRA utility with risk aversion bounded above by 1, lower interest rate increases consumption. \cite{Light2020} applies this result to prove the uniqueness of equilibrium in a certain Bewley-Aiyagari model. \cite{StachurskiToda2019JET, StachurskiToda2020Corrigendum} show that consumption functions have linear lower bounds when the relative risk aversion is bounded, which they apply to show that wealth inherits the tail behavior of income in general equilibrium models with labor income risk only.

With HARA preferences and general income shocks, \cite{CarrollKimball1996} show the concavity of consumption functions in finite horizon problems, which implies asymptotic linearity. However, under certain regularity assumptions, \cite{TodaHARA} shows that HARA is necessary for the concavity of consumption functions, implying that establishing asymptotic linearity based on concavity is possible only in very special cases. Furthermore, \cite{MaStachurskiToda2020JET} extend the concavity result of \cite{CarrollKimball1996} to infinite horizon and prove asymptotic linearity of the optimal consumption function. \cite{MaToda2021JET} characterize the asymptotic MPCs analytically in the framework of \cite{MaStachurskiToda2020JET} specialized to CRRA utility. As discussed above, out paper separates from these studies in that we impose significantly weaker assumptions on the utility function.

\section{Main results}\label{sec:main_res}
\subsection{Optimal savings problem}\label{sec:IF}
In this section we introduce a general optimal savings problem that closely follows the settings in \cite{MaStachurskiToda2020JET} and \cite{MaToda2021JET}. To avoid redundancies, we limit the discussion to the bare essentials. More details may be found  in \citet[Section 2.1]{MaToda2021JET}.

Time is discrete and denoted by $t=0,1,\dots,n$, with possibly $n=\infty$. Let $a_t$ be the financial wealth of the agent at the beginning of period $t$. The agent chooses consumption $c_t\ge 0$ and saves the remaining wealth $a_t-c_t$. The period utility function is $u$ and the discount factor, gross return on wealth, and non-financial income in period $t$ are denoted by $\beta_t,R_t,Y_t$, where we normalize $\beta_0=1$. Thus the agent solves
\begin{subequations}\label{eq:IF}
\begin{align}
&\maximize && \E_0\sum_{t=0}^n \left(\prod_{i=0}^t\beta_i\right)u(c_t) \notag\\
&\st && a_{t+1}=R_{t+1}(a_t-c_t)+Y_{t+1}, \label{eq:IF_budget}\\
&&& 0\le c_t\le a_t, \label{eq:IF_borrow}
\end{align}
\end{subequations}
where the initial wealth $a_0=a>0$ is given, \eqref{eq:IF_budget} is the budget constraint, and \eqref{eq:IF_borrow} implies that the agent cannot borrow. 
The stochastic processes $\set{\beta_t,R_t,Y_t}_{t\ge 1}$ obey
\begin{equation}
\beta_t=\beta(Z_{t-1},Z_t,\zeta_t),\quad R_t=R(Z_{t-1},Z_t,\zeta_t),\quad Y_t=Y(Z_{t-1},Z_t,\zeta_t),\label{eq:betaRY}
\end{equation}
where $\beta,R,Y$ are nonnegative measurable functions, $\set{Z_t}_{t\ge 0}$ is a time-homogeneous Markov chain taking values in a finite set $\ZZ=\set{1,\dots,Z}$ with a transition probability matrix $P$, and the innovations $\set{\zeta_t}$ are independent and identically distributed (\iid) over time and could be vector-valued.
To simplify the notation, we introduce the following conventions. We use a hat to denote a random variable that is realized next period, for example $Z=Z_t$ and $\hat{Z}=Z_{t+1}$. When no confusion arises, we write $\hat{\beta}$ for $\beta(Z,\hat{Z},\hat{\zeta})$ and define $\hat{R},\hat{Y}$ analogously. Conditional expectations are abbreviated using subscripts, for example
\begin{equation*}
\E_zX=\CE{X|Z=z}\quad \text{and}\quad \E_{z,\hat{z}}X=\CE{X|Z=z,\hat{Z}=\hat{z}}.
\end{equation*}
For $\theta\in \R$, we define the matrix $K(\theta)$ related to the transition probability matrix $P$, discount factor $\beta$, and return $R$ by
\begin{equation}
K_{z\hat{z}}(\theta)\coloneqq P_{z\hat{z}}\E_{z,\hat{z}}\hat{\beta}\hat{R}^\theta=P_{z\hat{z}}\E\beta(z,\hat{z},\hat{\zeta})R(z,\hat{z},\hat{\zeta})^\theta\in [0,\infty]. \label{eq:Ktheta}
\end{equation}
The matrix $K(\theta)$ for various values of $\theta$ appears throughout the paper. For a square matrix $A$, let $r(A)$ denote its spectral radius (largest absolute value of all eigenvalues).

Consider the following assumptions.

\begin{asmp}\label{asmp:Inada}
The utility function $u:[0,\infty)\to \R \cup \set{-\infty}$ is continuously differentiable on $(0,\infty)$, $u'$ is positive and strictly decreasing on $(0,\infty)$, and $u'(\infty)=0$.
\end{asmp}

Assumption~\ref{asmp:Inada} is essentially the usual monotonicity and concavity assumptions together with a form of Inada condition. %\footnote{Note that the condition $u'(\infty)<1$ is without loss of generality. This is because under the assumption that $u'$ is positive and strictly decreasing, we have $0<u'(\infty)<u'(1)<\infty$, and we may replace the utility function $u$ by its affine transformation $v\coloneqq u/u'(1)$, which satisfies $v'(\infty)=u'(\infty)/u'(1)<1$.}

\begin{asmp}\label{asmp:spectral}
Let $K$ be as in \eqref{eq:Ktheta}. The following conditions hold:
\begin{enumerate}
\item\label{item:Ebeta} The matrices $K(0)$ and $K(1)$ are finite,
\item\label{item:rbeta} If $n=\infty$, then $r(K(0))<1$ and $r(K(1))<1$,
\item\label{item:EY} $\E_{z,\hat{z}}\hat{Y}<\infty$, $\E_{z,\hat{z}}u'(\hat{Y})<\infty$, and $\E_{z,\hat{z}}\hat{\beta}\hat{R}u'(\hat{Y})<\infty$ for all $(z,\hat{z})\in \ZZ^2$.
\end{enumerate}
\end{asmp}
Under the maintained assumptions, Theorem~\ref{thm:exist} below states that the optimal savings problem \eqref{eq:IF} admits a unique solution and provides a computational algorithm. To make its statement precise, we introduce further definitions. Let $\cC$ be the space of candidate consumption functions such that $c:(0,\infty)\times \ZZ\to \R$ is continuous, is increasing in the first argument, $0\le c(a,z)\le a$ for all $a>0$ and $z\in \ZZ$, and
\begin{equation}
\sup_{(a,z)\in (0,\infty)\times \ZZ}\abs{u'(c(a,z))-u'(a)}<\infty.\label{eq:uprimediff}
\end{equation}
For $c,d\in \cC$, define the metric
\begin{equation}
\rho(c,d)=\sup_{(a,z)\in (0,\infty)\times \ZZ}\abs{u'(c(a,z))-u'(d(a,z))}.\label{eq:rho}
\end{equation}
When $u'$ is positive, continuous, and strictly decreasing (implied by Assumption~\ref{asmp:Inada}), it is straightforward (\eg, Proposition 4.1 of \cite*{LiStachurski2014}) to show that $(\cC,\rho)$ is a complete metric space.

\cite*{LiStachurski2014}, \cite{MaStachurskiToda2020JET}, and \cite{MaToda2021JET} show that the solution to the optimal savings problem \eqref{eq:IF} can be obtained as the unique fixed point of the policy iteration operator, which updates the consumption function using the Euler equation. The following lemma defines the updating rule. In what follows, long proofs are relegated to Appendix~\ref{sec:proof}.

\begin{lem}[\citealp{MaToda2021JET}, Lemma~1]\label{lem:xi}
Suppose that $u'$ is continuous, positive, strictly decreasing, and $\E_{z,\hat{z}}\hat{\beta}\hat{R}<\infty$ and $\E_{z,\hat{z}}u'(\hat{Y})<\infty$ for all $(z,\hat{z})\in \ZZ^2$. Then for any $c\in \cC$, $a>0$, and $z\in \ZZ$, there exists a unique $\xi\in [0,a]$ satisfying the Euler equation
\begin{equation}
u'(\xi)=\min\set{\max\set{\E_z \hat{\beta}\hat{R}u'(c(\hat{R}(a-\xi)+\hat{Y},\hat{Z})),u'(a)},u'(0)},\label{eq:xi}
\end{equation}
with $\xi>0$ if $u'(0)=\infty$.
\end{lem}

When Assumptions~\ref{asmp:Inada}, \ref{asmp:spectral} hold and $c\in \cC$, $a>0$, and $z\in \ZZ$, by Lemma~\ref{lem:xi} we can define a unique number $Tc(a,z)\coloneqq \xi\in [0,a]$ that solves \eqref{eq:xi}. The following theorem establishes the existence and uniqueness of a solution to the optimal savings problem \eqref{eq:IF}.

\begin{thm}\label{thm:exist}
Suppose Assumptions~\ref{asmp:Inada} and \ref{asmp:spectral} hold. Then $T$ is a monotone self map on $\cC$ and a unique solution $c\in \cC$ to the optimal savings problem \eqref{eq:IF} exists, which is characterized as follows:
\begin{enumerate}
\item\label{item:exist_finite} If $n<\infty$, then $c=T^nc_0$, where $c_0(a,z)\coloneqq a$.
\item\label{item:exist_inf} If $n=\infty$, then $c$ is the unique fixed point of $T$, and we have $T^nc_0\to c$ for any $c_0\in \cC$ (in particular, $c_0(a,z)\coloneqq a$).
\end{enumerate}
\end{thm}

\begin{proof}
The case $n=\infty$ is established in \citet[Theorem~2]{MaToda2021JET}. The case with finite $n$ follows by backward induction.
\end{proof}

\subsection{Asymptotic linearity of consumption functions}\label{sec:AL}

In this section we show that the consumption functions in the optimal savings problem \eqref{eq:IF} are asymptotically linear when the marginal utility is regularly varying. This is also where we depart from the setting in \cite{MaToda2021JET}, who assume outright that the utility function exhibits constant relative risk aversion (CRRA). To state the main results, we introduce several notions. A positive measurable function $\ell:(0,\infty)\to (0,\infty)$ is \emph{slowly varying} if $\ell(\lambda x)/\ell(x)\to 1$ as $x\to \infty$ for all $\lambda>0$. A positive measurable function $f:(0,\infty)\to (0,\infty)$ is \emph{regularly varying} with index $\rho\in \R$ if $f(\lambda x)/f(x)\to \lambda^\rho$ as $x\to \infty$ for all $\lambda>0$. \citet[Theorem 1.4.1]{BinghamGoldieTeugels1987} show that if $f$ is a positive measurable function such that $g(\lambda)=\lim_{x\to\infty} f(\lambda x)/f(x)\in (0,\infty)$ exists for $\lambda>0$ in a set of positive measure, then the limit function must be of the form $g(\lambda)=\lambda^\rho$ for some $\rho\in \R$ and $f$ is regularly varying with index $\rho$.

%We assume that marginal utility is regularly varying.

\begin{asmp}\label{asmp:regular}
The marginal utility function is regularly varying with index $-\gamma<0$; equivalently, there exists a slowly varying function $\ell$ such that $u'(c)=c^{-\gamma}\ell(c)$.
\end{asmp}

The assumption that marginal utility is regularly varying is related to several assumptions in the literature. Following \cite{BrockGale1969} and \cite{SchechtmanEscudero1977}, we say that $u'$ has an \emph{asymptotic exponent} $-\gamma$ if $\log u'(c)/\log c\to -\gamma$ as $c\to\infty$. We say that $u$ is \emph{asymptotically CRRA} with coefficient $\gamma$ if $u$ is twice differentiable and $-cu''(c)/u'(c)\to \gamma$ as $c\to\infty$. The following proposition clarifies the relation between these concepts.

\begin{prop}\label{prop:RV}
Let $\aCRRA(\gamma)$, $\RV(-\gamma)$, and $\aE(-\gamma)$ be respectively the class of utility functions $u$ such that $u$ is asymptotically CRRA with coefficient $\gamma$, $u'$ is regularly varying with index $-\gamma$, and $u'$ has an asymptotic exponent $-\gamma$. Then
\begin{equation*}
\aCRRA(\gamma)\subsetneq \RV(-\gamma)\subsetneq \aE(-\gamma).
\end{equation*}
\end{prop}

It is clear from Proposition \ref{prop:RV} that Assumption \ref{asmp:regular} is significantly weaker than CRRA (assumed in \citealp{MaToda2021JET}) because it imposes a parametric assumption only at infinity ($c\to\infty$). Furthermore, the parameter $\gamma>0$ can be interpreted as the asymptotic relative risk aversion of the agent.

To prove our main results, we introduce a technical condition that permits us to apply the dominated convergence theorem.

\begin{asmp}\label{asmp:R}
There exists $\delta>0$ such that $R(z,\hat{z},\zeta)\in \set{0}\cup [\delta,\infty)$ almost surely conditional on all $(z,\hat{z})\in \ZZ^2$.
\end{asmp}

Assumption~\ref{asmp:R} holds, for example, if the \iid innovation $\zeta$ takes finitely many values, which is almost always the case in applied numerical works that employ discretization. Note that we allow the possibility $R=0$ with positive probability. Throughout the rest of the paper, we introduce the following conventions to simplify the notation: ``1'' denotes either the real number 1 or the vector $(1,\dots,1)'\in \R^Z$ depending on the context; we interpret $0\cdot \infty=0$ and $\beta R^{1-\gamma}=(\beta R)R^{-\gamma}$, so $\beta R^{1-\gamma}=0$ whenever $\beta=0$ or $R=0$ regardless of the value of $\gamma>0$. Although the following property is an immediate implication of the above assumptions and convention, we state it as a lemma since we frequently refer to it.

\begin{lem}\label{lem:EbetaR}
Suppose Assumptions \ref{asmp:spectral}\ref{item:Ebeta} and \ref{asmp:R} hold. Then $\E_{z,\hat{z}}\beta R^{1-\gamma}<\infty$ for all $(z,\hat{z})\in \ZZ^2$. Consequently, the matrix $K(1-\gamma)$ defined in \eqref{eq:Ktheta} is finite.
\end{lem}

\begin{proof}
If $R=0$, then $\beta R^{1-\gamma}=(\beta R)R^{-\gamma}=0$ by convention. If $R>0$, then $R\ge \delta$ almost surely by Assumption \ref{asmp:R}. In either case $\beta R^{1-\gamma}\le \beta R\delta^{-\gamma}$, so
$$\E_{z,\hat{z}}\beta R^{1-\gamma}\le \E_{z,\hat{z}}\beta R\delta^{-\gamma}=\delta^{-\gamma}\E_{z,\hat{z}}\beta R<\infty$$
by Assumption \ref{asmp:spectral}\ref{item:Ebeta}.
\end{proof}

Under the maintained assumptions, we can show that the consumption functions are asymptotically linear, which is our main result. We present two results, one for finite horizon ($n<\infty$) and another for infinite horizon ($n=\infty$).

\begin{thm}[Asymptotic linearity, $n<\infty$]\label{thm:AL_finite}
Suppose Assumptions \ref{asmp:Inada}--\ref{asmp:R} hold. Define the map $F:\R_+^Z\to \R_+^Z$ and sequence $\set{x_n}_{n=0}^\infty\subset \R_+^Z$ by
\begin{equation}
(Fx)(z)\coloneqq \left(1+(K(1-\gamma)x)(z)^{1/\gamma}\right)^\gamma, \quad z=1,\dots,Z, \label{eq:F}
\end{equation}
$x_0=1$, and $x_n=Fx_{n-1}$ for all $n\in \N$. Let $c_n(a,z)$ be the $n$-period consumption function established in Theorem~\ref{thm:exist}\ref{item:exist_finite}. Then
\begin{equation}
\lim_{a\to\infty}\frac{c_n(a,z)}{a}=x_n(z)^{-1/\gamma}.\label{eq:cnbar}
\end{equation}
\end{thm}

\begin{thm}[Asymptotic linearity, $n=\infty$]\label{thm:AL_inf}
Let everything be as in Theorem~\ref{thm:AL_finite} and $c(a,z)$ be the consumption function established in Theorem \ref{thm:exist}\ref{item:exist_inf}. Then the sequence $\set{x_n}_{n=0}^\infty\subset \R_+^Z$ monotonically converges to some $x^*\in (0,\infty]^Z$ and
\begin{equation}
0\le \liminf_{a\to\infty}\frac{c(a,z)}{a}\le \limsup_{a\to\infty}\frac{c(a,z)}{a}\le x^*(z)^{-1/\gamma}.\label{eq:MPC_bound}
\end{equation}
Furthermore, the following statements are true.
\begin{enumerate}
\item\label{item:rless1} If $r(K(1-\gamma))<1$, then $x^*$ is the unique fixed point of $F$ in \eqref{eq:F}. If in addition $\liminf_{a\to\infty}c(a,z)/a>0$ for all $z\in \ZZ$, then
\begin{equation}
\bar{c}(z)\coloneqq \lim_{a\to\infty}\frac{c(a,z)}{a}=x^*(z)^{-1/\gamma}\in (0,1].\label{eq:MPC_lim}
\end{equation}
\item\label{item:rgtr1} If $r(K(1-\gamma))\ge 1$, then $F$ in \eqref{eq:F} has no fixed point and $x^*(z)=\infty$ for some $z$. If in addition $K(1-\gamma)$ is irreducible, then for all $z\in \ZZ$ we have
\begin{equation*}
\bar{c}(z)\coloneqq \lim_{a\to\infty}\frac{c(a,z)}{a}=0.
\end{equation*}
\end{enumerate}
\end{thm}

A few remarks are in order. First, when the marginal utility is regularly varying, under the stated technical conditions, Theorem~\ref{thm:AL_finite} shows that consumption functions in finite horizon problems are always asymptotically linear, and the asymptotic MPCs are exactly characterized as in \eqref{eq:cnbar}.

Second, one might conjecture that similar statements hold in infinite horizon problems by taking the limit of both sides of \eqref{eq:cnbar}, but this is not generally true. The reason is that we cannot interchange the two limits $a\to\infty$ and $n\to\infty$. To see this, consider the function $f_n:[0,\infty)\to \R$ defined by $f_n(a)=\max\set{a-n,0}$. Then clearly $\lim_{n\to\infty} f_n(a)=0\eqqcolon f(a)$ pointwise and $\lim_{a\to\infty}f_n(a)/a=1$, so
\begin{equation*}
\lim_{n\to\infty}\lim_{a\to\infty} \frac{f_n(a)}{a}=1\neq 0=\lim_{a\to\infty}\frac{f(a)}{a}=\lim_{a\to\infty}\lim_{n\to\infty} \frac{f_n(a)}{a}.
\end{equation*}
This observation explains why in general we can only obtain bounds of the form \eqref{eq:MPC_bound}.

Third, an immediate implication of Theorem \ref{thm:AL_inf} is that, when $r(K(1-\gamma))<1$, either $\liminf_{a\to\infty}c(a,z)/a=0$ for some $z\in \ZZ$, or \eqref{eq:MPC_lim} holds for all $z\in \ZZ$. However, the theorem does not tell which case occurs, which we investigate in the next section.

\subsection{Sufficient conditions}\label{sec:suff_cond}

We now seek sufficient conditions for the limit \eqref{eq:MPC_lim} to hold. Throughout this section, we maintain Assumptions~\ref{asmp:Inada}--\ref{asmp:R} (with $n=\infty$) and assume $r(K(1-\gamma))<1$. We first present a general result that relies on high-level assumptions, followed by applications to specific cases.

\begin{thm}\label{thm:suff_cond}
Let $b=\inf Y\ge 0$. If there exist numbers $\set{\epsilon(z)}_{z\in \ZZ} \subset (0,1)$ and $A\ge 0$ such that
\begin{equation}
\delta \left(1 - \max_{z\in \ZZ} \epsilon(z) \right)A+b \ge A\label{eq:Acond}
\end{equation}
and
\begin{equation}
u'(\epsilon(z)a)\ge \E_z\hat{\beta}\hat{R}u'\left(\epsilon(\hat{Z})\hat{R}[1-\epsilon(z)]a \right)
\label{eq:suff_cond1}
\end{equation}
for all $a>A$ and $z\in \ZZ$, then $c(a,z)\ge \epsilon(z)a$ for all $a>A$ and $z\in \ZZ$. In particular, the limit \eqref{eq:MPC_lim} holds by Theorem~\ref{thm:AL_inf}\ref{item:rless1}.
\end{thm}

To obtain the limit \eqref{eq:MPC_lim}, it thus suffices to verify the conditions \eqref{eq:Acond} and \eqref{eq:suff_cond1}. To this end, we rewrite \eqref{eq:suff_cond1} in a more convenient form. In what follows, assume $u$ is twice continuously differentiable and let
\begin{equation*}
\gamma(c)\coloneqq -\frac{cu''(c)}{u'(c)}\ge 0
\end{equation*}
be the local relative risk aversion coefficient.

Take any numbers $\set{\epsilon(z)}_{z\in \ZZ}\subset (0,1)$ and define $x=\epsilon(z)a$ and $y=\epsilon(\hat{Z})\hat{R}[1-\epsilon(z)]a$. By the mean value theorem for integrals, we obtain
\begin{align*}
\log \frac{u'(y)}{u'(x)}&=\int_x^y(\log u'(c))'\diff c=\int_x^y\frac{u''(c)}{u'(c)}\diff c\\
&=-\int_x^y\frac{\gamma(c)}{c}\diff c=-\int_x^y\frac{\hat{\gamma}}{c}\diff c=-\hat{\gamma}\log \frac{y}{x},
%\\
%\implies \frac{u'(y)}{u'(x)}&=\left(\frac{y}{x}\right)^{-\hat{\gamma}},
\end{align*}
where $\hat{\gamma}=\gamma(\hat{c})$ for $\hat{c}=(1-\theta)x+\theta y$ with some $\theta\in (0,1)$. Taking the exponential of both sides, we obtain $u'(y)/u'(x)=(y/x)^{-\hat{\gamma}}$. Therefore
\begin{equation}
\frac{\E_z \hat{\beta}\hat{R}u'(\epsilon(\hat{Z})\hat{R}[1-\epsilon(z)]a)}{u'(\epsilon(z)a)}=\E_z\hat{\beta}\hat{R}\left(\frac{\epsilon(\hat{Z})}{\epsilon(z)}\hat{R}[1-\epsilon(z)]\right)^{-\hat{\gamma}}.\label{eq:eulerRatio}
\end{equation}
Thus, verifying \eqref{eq:suff_cond1} reduces to checking that the right-hand side of \eqref{eq:eulerRatio} is at most 1. Note that the right-hand side depends on $u$ and $a$ only through $\hat{\gamma}$.

We now present several sufficiency results.

\begin{prop}[Constant relative risk aversion]\label{prop:CRRA}
If $u$ exhibits constant relative risk aversion (CRRA), so $u'(c)=c^{-\gamma}$, then the limit \eqref{eq:MPC_lim} holds.
\end{prop}

\begin{proof}
This is Theorem 3 of \cite{MaToda2021JET}, which holds under weaker assumptions (Assumption~\ref{asmp:R} can be dropped). It is also immediate from Theorem~\ref{thm:suff_cond} by setting $A=0$ (which implies \eqref{eq:Acond}), $\epsilon(z)=x^*(z)^{-1/\gamma}$, and noting that the right-hand side of \eqref{eq:eulerRatio} equals 1 because $\hat{\gamma}=\gamma$ and $x^*$ is the fixed point of $F$ in \eqref{eq:F} (see \eqref{eq:xstarEq}).
\end{proof}

\begin{prop}[Bounded relative risk aversion]\label{prop:BRRA}
If $u$ exhibits bounded relative risk aversion (BRRA), so
\begin{equation*}
0\le \ubar{\gamma}\coloneqq \inf_{c>0}\gamma(c)\le \sup_{c>0}\gamma(c)\eqqcolon \bar{\gamma}<\infty,
\end{equation*}
and
\begin{equation}
\max_{z\in \ZZ}\E_z\hat{\beta}\hat{R}\max\set{\hat{R}^{-\ubar{\gamma}},\hat{R}^{-\bar{\gamma}}}<1,\label{eq:suff_cond2}
\end{equation}
then the limit \eqref{eq:MPC_lim} holds.
\end{prop}

\begin{proof}
Take $A=0$, which implies \eqref{eq:Acond}. We aim to show \eqref{eq:suff_cond1} for $\epsilon(z)\equiv \epsilon$ (constant) with sufficiently small $\epsilon>0$. Since by assumption $\hat{\gamma}\in [\ubar{\gamma},\bar{\gamma}]$, it follows from \eqref{eq:eulerRatio} that
\begin{align*}
\frac{\E_z\hat{\beta}\hat{R}u'(\epsilon \hat{R}(1-\epsilon)a)}{u'(\epsilon a)}&\le \E_z\hat{\beta}\hat{R}\max\set{[\hat{R}(1-\epsilon)]^{-\ubar{\gamma}},[\hat{R}(1-\epsilon)]^{-\bar{\gamma}}}\\
&\to \E_z\hat{\beta}\hat{R}\max\set{\hat{R}^{-\ubar{\gamma}},\hat{R}^{-\bar{\gamma}}}<1
\end{align*}
as $\epsilon\downarrow 0$ by \eqref{eq:suff_cond2}. Therefore \eqref{eq:suff_cond1} holds for $\epsilon(z)\equiv \epsilon$ (constant) with sufficiently small $\epsilon>0$.
\end{proof}

In many applied works, it is often the case that the discount factor $\beta$ is constant and agents invest only in a risk-free asset. In this case we obtain the following corollary.

\begin{cor}[Constant $\beta,R$]\label{cor:BRRA}
If $u$ is BRRA, $\beta,R$ are constant, and $R\ge 1$, then the limit \eqref{eq:MPC_lim} holds. Furthermore,
\begin{equation}
\bar{c}(z)=1-(\beta R^{1-\gamma})^{1/\gamma}.\label{eq:cbar_cons}
\end{equation}
\end{cor}

\begin{proof}
Since $R\ge 1$, we have $R^{-\bar{\gamma}}\le R^{-\ubar{\gamma}}\le 1$. Therefore
\begin{equation*}
\max_{z\in \ZZ}\E_z\hat{\beta}\hat{R}\max\set{\hat{R}^{-\ubar{\gamma}},\hat{R}^{-\bar{\gamma}}}\le \beta R=r(K(1))<1
\end{equation*}
by Assumption~\ref{asmp:spectral}\ref{item:rbeta}, so \eqref{eq:suff_cond2} holds. The expression for $\bar{c}(z)$ follows from Example 2 of \cite{MaToda2021JET}.
\end{proof}

Under the assumptions of Corollary~\ref{cor:BRRA}, Proposition 5 of \cite{StachurskiToda2019JET} (which has been corrected as Proposition 5' of \cite{StachurskiToda2020Corrigendum}) establishes that consumption functions have linear lower bounds. Corollary~\ref{cor:BRRA} strengthens their result because it proves the asymptotic linearity with an exact characterization of the asymptotic MPC.

\begin{prop}[Asymptotic CRRA]\label{prop:aCRRA}
If $u$ exhibits asymptotically constant relative risk aversion (aCRRA), so $\gamma(c)\to \gamma$ as $c\to\infty$, and $b=\inf Y\ge 0$ is large enough, then the limit \eqref{eq:MPC_lim} holds. If in addition $\delta>1$ in Assumption~\ref{asmp:R}, then the conclusion holds for any $b\ge 0$.
\end{prop}

Although Proposition~\ref{prop:aCRRA} does not provide an explicit threshold for the minimum income $b=\inf Y$ so that the limit \eqref{eq:MPC_lim} holds, it is clear from its proof that the threshold can be calculated if the utility function $u$ is explicitly given. In particular, how small $b$ can be depends on how fast the local relative risk aversion $\gamma(c)$ converges to $\gamma$.

\section{Computational efficiency}\label{sec:comp}

%Theorems~\ref{thm:AL_finite}-\ref{thm:AL_inf} and Propositions~\ref{prop:CRRA}--\ref{prop:aCRRA} have several implications. First, they justify the textbook affine consumption function as in \cite{keynes1936}. Second, they explain the ``approximate aggregation'' property in heterogeneous-agent general equilibrium models as in \cite{krusell-smith1998}. Approximate aggregation refers to the observation that, when solving heterogeneous-agent general equilibrium models, keeping track of just the first moment of the wealth distribution is nearly sufficient, despite the fact that the entire wealth distribution is a state variable. Because the market clearing condition involves aggregate savings, aggregation would be possible if saving is linear in wealth. Our results show that consumption (hence saving) is approximately linear in wealth, which explains the approximate aggregation property. Finally, they justify linearly extrapolating policy functions outside the grid when solving models numerically as in \cite{Gouin-BonenfantTodaParetoExtrapolation}.

In this section, we discuss the computational aspects of the optimal savings problem based on the theory we derive. In principle, given Theorem~\ref{thm:exist}\ref{item:exist_inf}, one can compute the solution $c\in \cC$ to the optimal savings problem \eqref{eq:IF} by starting from any $c_0\in \cC$ and iterating the policy iteration operator $T$. However, in practice there are many fine details that need to be addressed. We divide our discussion into initializing $c$ and updating $c$.

\subsection{Initializing $c(a,z)$}

Let $c\in \cC$ be the solution to the optimal savings problem \eqref{eq:IF} and $T:\cC\to \cC$ be the policy iteration operator defined in Section~\ref{sec:IF}. Theorem 2.2 of \cite{MaStachurskiToda2020JET}, which also holds in the more general settings in \cite{MaToda2021JET} and Section~\ref{sec:IF}, shows that $T^k$ is a contraction for some $k\in \N$. Consequently, by the contraction mapping theorem, letting $\rho$ be the marginal utility metric in \eqref{eq:rho}, there exists a number $r\in (0,1)$ such that the approximation error can be bounded as
\begin{equation*}
\rho(T^{kn} c_0,c)\le r^n \rho(c_0,c)\to 0~\text{as}~n\to\infty
\end{equation*}
for any initial guess $c_0\in \cC$. Therefore to compute $c$ efficiently, it is important to start with a good initial guess $c_0\in \cC$.

If Assumptions~\ref{asmp:Inada}--\ref{asmp:R} hold and $r(K(1-\gamma))<1$, we know from Theorem~\ref{thm:AL_inf}\ref{item:rless1} that the $c(a,z)/a$ is asymptotically bounded above by $x^*(z)^{-1/\gamma}$, where $x$ is the unique fixed point of $F$ in \eqref{eq:F}. Therefore it is reasonable to use an affine function with slope $\bar{c}(z)$ as an initial guess of the consumption function $c(a,z)$. To determine its intercept, we do as follows. Let $\bar{a}(z)$ be the asset threshold for the borrowing constraint to bind (so $c(a,z)=a$ for $a\le \bar{a}(z)$ and $c(a,z)<a$ for $a>\bar{a}(z)$). Then at $a=\bar{a}(z)$, the Euler equation \eqref{eq:xi} implies that
\begin{equation*}
u'(\bar{a}(z))=\E_z \hat{\beta}\hat{R}u'(c(\hat{R}(\bar{a}(z)-\bar{a}(z))+\hat{Y},\hat{Z}))=\E_z \hat{\beta}\hat{R}u'(c(\hat{Y},\hat{Z})).
\end{equation*}
Using the approximation $c(a,z)\approx a$ for small asset level, we obtain
\begin{equation}
u'(\bar{a}(z))\approx \E_z \hat{\beta}\hat{R}u'(\hat{Y})\iff \bar{a}(z)\approx (u')^{-1}\left(\E_z \hat{\beta}\hat{R}u'(\hat{Y})\right). \label{eq:abar}
\end{equation}
Therefore a reasonable initial guess based on theory is
\begin{align}
c_0(a,z)&\coloneqq \min\set{a,\bar{c}(z)(a-\bar{a}(z)) + \bar{a}(z)}\notag \\
&=\min\set{a,\bar{c}(z)a + (1-\bar{c}(z))\bar{a}(z)}, \label{eq:c0}
\end{align}
where $\bar{c}(z)=x^*(z)^{-1/\gamma}$ and $\bar{a}(z)$ is defined by the right-hand side of \eqref{eq:abar}. (In \eqref{eq:c0}, we take the minimum with $a$ to satisfy $c_0(a,z)\le a$.) This $c_0$ trivially belongs to the candidate space $\cC$. Furthermore, it satisfies
\begin{equation*}
\lim_{a\to\infty}\frac{c_0(a,z)}{a}=x^*(z)^{-1/\gamma}=\bar{c}(z)=\lim_{a\to\infty}\frac{c(a,z)}{a},
\end{equation*}
so we can expect that $c_0$ approximates $c$ well.

One remaining practical issue is how to numerically compute $x^*(z)$, which appears in \eqref{eq:c0}. By Theorem~\ref{thm:AL_inf}\ref{item:rless1}, $x^*$ is the unique positive solution to the equation $x=Fx$, which in principle can be computed using a nonlinear equation solver. However, doing so is not practical because $x^*$ tends to be a very large vector and nonlinear equation solvers tend to terminate before convergence is achieved. For instance, suppose $\beta,R$ are constant and $1\le R<1/\beta$. Using \eqref{eq:cbar_cons}, we obtain
\begin{equation*}
x^*(z)=\bar{c}(z)^{-\gamma}=(1-(\beta R^{1-\gamma})^{1/\gamma})^{-\gamma}.
\end{equation*}
Suppose we fix the unit of time somehow (\eg, year), one period has time length $\Delta$, and (with a slight abuse of notation) the discount rate is $\delta>0$ and the continuously compounded risk-free rate is $r\in (0,\delta)$. Then $\beta=\e^{-\delta\Delta}$ and $R=\e^{r\Delta}$, so
\begin{equation}
x^*(z)=\left(1-\e^{-\frac{\Delta}{\gamma}(\delta-(1-\gamma)r)}\right)^{-\gamma}. \label{eq:xzDelta}
\end{equation}
It is clear that \eqref{eq:xzDelta} diverges to $\infty$ as $\Delta\downarrow 0$. In fact, \eqref{eq:xzDelta} tends to be quite large in common settings. For instance, let the unit of time be a year and $\delta=0.04$ (4\% annual discounting), $r=0.03$ (3\% risk-free rate), and $\gamma=3$, which are typical values. Figure~\ref{fig:xz} plots the fixed point \eqref{eq:xzDelta} in the range $\Delta\in [1/12,1]$, which implies that a period is between a month and a year. We can see that $x^*(z)$ tends to be quite large, which makes it impractical to numerically solve for $x^*(z)$.

\begin{figure}[!htb]
\centering
\includegraphics[width=0.6\linewidth]{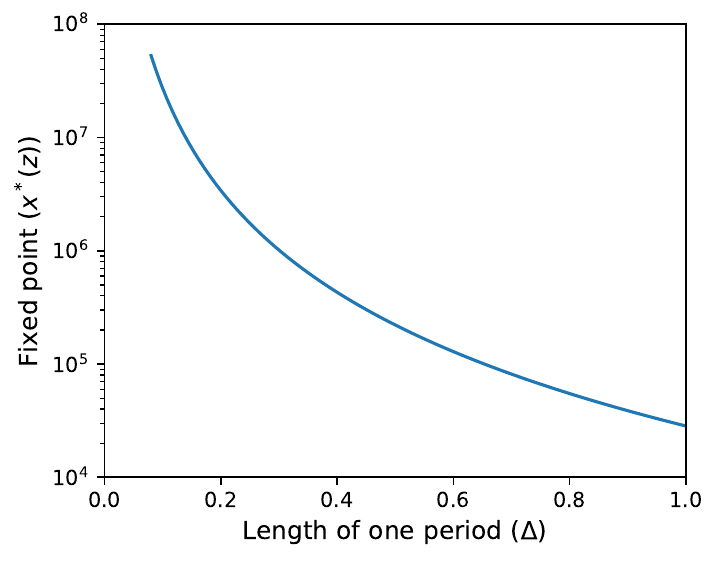}
\caption{Fixed point $x^*(z)$ in \eqref{eq:xzDelta} with $(\gamma,\delta,r)=(3,0.04,0.03)$.}\label{fig:xz}
\end{figure}

A straightforward way to avoid this issue is to directly solve for $\bar{c}(z)=x^*(z)^{-1/\gamma}\in (0,1)$ instead of $x^*(z)$. Noting that $x^*$ is the fixed point of $F$ in \eqref{eq:F}, $\bar{c}(z)$ satisfies
\begin{equation}
\bar{c}(z)=\left(1+\left(\sum_{\hat{z}=1}^Z K_{z\hat{z}}(1-\gamma)\bar{c}(\hat{z})^{-\gamma}\right)^{1/\gamma}\right)^{-1},\quad z=1,\dots,Z.\label{eq:cbareq}
\end{equation}
To numerically solve the system of equations \eqref{eq:cbareq} using a nonlinear equation solver, we may use a good initial guess as follows. Conjecture that $x^*$ is approximately equal to $k1$ for some $k>0$ and that the right Perron vector of $K(1-\gamma)$ is close to 1. Then the equation $x=Fx$ becomes
\begin{equation*}
k1\approx \left(1+(K(1-\gamma)k1)^{1/\gamma}\right)^\gamma\approx \left(1+(r (K(1-\gamma))k)^{1/\gamma}\right)^\gamma 1.
\end{equation*}
Solving for $k$, we obtain
\begin{equation}
\bar{c}(z)\approx k^{-1/\gamma}=1-r(K(1-\gamma))^{1/\gamma}.\label{eq:cbarapprox}
\end{equation}
Since $r(K(1-\gamma))<1$ by assumption, the right-hand side of \eqref{eq:cbarapprox} is always in $(0,1)$. Therefore we can numerically solve the system of nonlinear equations \eqref{eq:cbareq} using the initial guess \eqref{eq:cbarapprox}.

\subsection{Updating $c(a,z)$}

Given a candidate consumption function $c\in \cC$, it is natural to update it to $Tc$ using the Euler equation \eqref{eq:xi}. However, this is not generally feasible because $c$ is a function and thus $\cC$ is infinite-dimensional. In practice, it is common to set up a finite grid $\cA_G\coloneqq \set{a_g}_{g=1}^G$, where $0<a_1<\dots<a_G$ are grid points, update $c$ by solving for $\xi = c(a_g,z)$ using \eqref{eq:xi} for each $(a_g,z) \in \cA_G\times \ZZ$, and interpolate (linear, spline, etc.) it if necessary. 

%For concreteness, suppose we use linear interpolation and extrapolation. Let $\cC(\cA_G)$ be the set of continuous piecewise linear functions $c:(0,\infty)\times \ZZ\to \R$ such that for all $z\in \ZZ$,
%\begin{enumerate*}
%\item $0<c(a,z)\le a$ for all $a>0$,
%\item $c(a,z)=a$ for $0<a\le a_1$,
%\item $c(a,z)$ is affine in $a$ on each subinterval $[a_g,a_{g+1}]$ for $g=1,\dots, G-1$, and
%\item $c(a,z)$ is linearly extrapolated for $a>a_G$.
%\end{enumerate*}
%Then clearly $\cC(\cA_G)\subset \cC_1$, where $\cC_1\subset \cC$ is the space of candidate consumption functions that are asymptotically linear defined in \eqref{eq:C1}. Furthermore, $\cC(\cA_G)$ is finite-dimensional because it is parametrized by $GZ$ numbers $\set{c(a_g,z)}_{g=1}^G{_{z=1}^Z}$. Thus, given $c\in \cC(\cA_G)$, one way to update it by solving for $\xi=c(a_g,z)$ using the Euler equation \eqref{eq:xi} for each $(a_g,z)\in \cA_G\times \ZZ$.

Although this updating procedure is theoretically justified, it is not computationally efficient because it requires performing root-finding $GZ$ times for each iteration, which is computationally intensive. A straightforward way to improve the algorithm is to avoid root-finding as in the endogenous grid point method of \cite{Carroll2006}. 

Instead of fixing a grid for asset, the endogenous grid method fixes a grid for saving $\cS_G \coloneqq \set{s_g}_{g=1}^G$, where %$0 < s_1 < \dots < s_G$,
$0=s_1<\dots<s_G$, update $c$ on $\cS_G \times \ZZ$, and then choose the grid points for asset endogenously based on the optimal consumption and saving. For concreteness, suppose we use linear interpolation and extrapolation. For each $n \in \N$, let $\cA_G^n \coloneqq \set{a_g^n(z)}_{g=1}^G{_{z=1}^Z}$ be the endogenous grid points for asset determined in the $n$-th iteration, where $a_g^n(z)$ represents the asset grid point when saving is $s_g$ and exogenous state is $z$. We define $\cA_G^0 = \set{a_g^0 (z)}_{g=1}^G{_{z=1}^Z}$ by $a_g^0(z) = s_g$ for all $z \in \ZZ$ and $g = 1, \dots, G$. Furthermore, let $\cC(\cA_G^n)$ be the set of continuous piecewise linear functions $c : (0, \infty) \times \ZZ \to \R$ such that for each $z \in \ZZ$, 
\begin{enumerate*}
\item $0<c(a,z)\le a$ for all $a>0$,
\item $c(a,z)=a$ for $0<a\le a_1^n(z)$,%$0<a\le s_1$,
\item $c(a,z)$ is affine in $a$ on each subinterval $[a_g^n(z),a_{g+1}^n(z)]$ for $g=1,\dots, G-1$, and
\item $c(a,z)$ is linearly extrapolated for $a>a_G^n(z)$.
\end{enumerate*} 
Policy iteration for computing the consumption functions via the endogenous grid method can be summarized as follows.

\begin{framed}
\begin{oneshot}[Policy iteration via the endogenous grid method]\label{alg:politer}
\quad 
\begin{enumerate}[1]
\item\label{item:initial} (Initialization)
	\begin{enumerate}[(i)]
	\item Solve the system of nonlinear equations \eqref{eq:cbareq} using the initial guess \eqref{eq:cbarapprox}.
	
	\item Define $c_0\in \cC(\cA_G^0)$ by \eqref{eq:c0}.
	\end{enumerate}

\item\label{item:update} (Updating) For each $n\in \N$, given $c_{n-1}\in \cC(\cA_G^{n-1})$, update it as follows:
	\begin{enumerate}[(i)]
	\item\label{item:update1} For each $(s_g,z)\in \cS_G\times \ZZ$, compute $\hat{a}\coloneqq \hat{R} s_g + \hat{Y}$ and define
	\begin{equation*}
	\tilde c_n(s_g, z)\coloneqq
	(u')^{-1}\left(\min\set{\E_z \hat{\beta}\hat{R}u'(c_{n-1}(\hat{a},\hat{Z})),u'(0)}\right),
	\end{equation*}
	where $c_{n-1}(\hat{a},\hat{Z})$ is computed by linearly interpolating and extrapolating $c_{n-1}\in \cC(\cA^{n-1}_G)$.
	
	\item\label{item:update2} Define the updated asset grid $\cA_G^{n} = \set{a_g^n(z)}_{g=1}^G{_{z=1}^Z}$ and optimal consumption on $\cA_G^{n}$ by
	\begin{equation*}
	    a_g^n(z) \coloneqq s_g + \tilde c_n(s_g, z)
	    \quad \text{and} \quad
	    c_n (a_g^n(z),z) \coloneqq \tilde c_n(s_g, z).
	\end{equation*}
	
	\item\label{item:update3} %Reset $a_1^n(z) = c_n(a_1^n(z), z) = s_1$ for each $z \in \ZZ$, and define 
	Define $c_n\in \cC(\cA_G^n)$ by linearly interpolating and extrapolating using $\set{ c_n (a_g^n(z),z)}_{g=1}^G{_{z=1}^Z}$ and setting $c_n(a,z)=a$ for $a\le a_1^n(z)$.
	\end{enumerate}
\item\label{item:converge} (Convergence) Repeat Step \ref{item:update} over $n\in \N$ until the $GZ$ numbers $\set{c_n(a_g^n(z),z)}_{g=1}^G{_{z=1}^Z}$ converge.
\end{enumerate}
\end{oneshot}
\end{framed}

Note that we avoid the root-finding routine in Step \ref{item:update}\ref{item:update1} due to the endogenous grid selection. %Moreover, in Step~\ref{item:update}\ref{item:update3}, we reset the first asset grid point at each exogenous state to $s_1$ to improve interpolation.

\subsection{Numerical example}

To illustrate the computational efficiency of policy iteration, we solve a numerical example in this section.

\paragraph{Model specification}
The agent has CRRA utility with constant discount factor $\beta>0$ and relative risk aversion $\gamma>0$. There are only two states, so $\ZZ=\set{1,2}$, which we interpret as expansion and recession. Letting $Z_t$ be the Markov state at time $t$, labor income is $Y_t=Y(Z_t)\e^{gt}$, where $g$ is the growth rate of the trend. Suppose the agent invests a constant fraction of wealth $\theta\in (0,1)$ into a risky asset whose return is lognormal (with conditional log mean and volatility depending only on the current state $Z_t$) and the remaining fraction $1-\theta$ into a risk-free asset. Therefore, the asset return is
\begin{equation}
R(Z_{t-1},Z_t,\zeta_t)=R_f(\theta \exp(\mu(Z_t)+\sigma(Z_t)\zeta_t)+1-\theta),\label{eq:Rtheta}
\end{equation}
where $\mu(Z_t)$ and $\sigma(Z_t)$ are the conditional log risk premium and volatility, $\zeta_t\sim \iid N(0,1)$, and $R_f>0$ is the gross risk-free rate. Although our theory requires a stationary income process, due to homotheticity it is straightforward to allow for a trend. After simple algebra (\eg, Section 2.2 of \citealp*{Carroll2021}), instead of \eqref{eq:betaRY}, it suffices to use
\begin{subequations}\label{eq:betaRYtilde}
\begin{align}
\tilde{\beta}_t&=\beta \e^{(1-\gamma)g},\\
\tilde{R}_t&=R(Z_{t-1},Z_t,\zeta_t)\e^{-g},\\
\tilde{Y}_t&=Y_t\e^{-gt}=Y(Z_t),
\end{align}
\end{subequations}
which are stationary.

We set the parameters as follows. We suppose that one period is a month and set $\beta=\e^{-0.04/12}$ (4\% annual discounting) and $\gamma=3$, which are standard. For the Markov state $Z_t$, we use the 1947-2019 NBER recession indicator\footnote{\url{https://fred.stlouisfed.org/series/USREC}} and estimate the transition probability matrix
\begin{equation*}
P=\begin{bmatrix}
0.9854 & 0.0146\\
0.0902 & 0.9098
\end{bmatrix}
\end{equation*}
from the mean duration of expansions and recessions. For the asset return in \eqref{eq:Rtheta}, we set $\theta=0.6$, which is close to the calibrated value in \cite{MaToda2021JET}, and use the spreadsheet of \cite{welch-goyal2008}\footnote{\url{http://www.hec.unil.ch/agoyal/docs/PredictorData2019.xlsx}} to construct real log returns and estimate $\log R_f=5.251\times 10^{-4}$ (annual rate 0.63\%), $(\mu(1),\mu(2))=10^{-3}\times (6.8111,-1.7201)$, and $(\sigma(1),\sigma(2))=(0.0383, 0.0559)$. The growth rate $g=1.6213\times 10^{-3}$ is computed from the real per capita GDP growth.\footnote{\url{https://fred.stlouisfed.org/series/A939RX0Q048SBEA}} Finally, we set $(Y(1),Y(2))=(1,0.5)$ to make the graphs stand out. For computational purposes, we discretize the \iid shock $\zeta$ using the 7-point Gauss-Hermite quadrature.

\paragraph{Consumption functions and asymptotic MPCs}
We solve for consumption functions by policy iteration on a 1,000-point exponential grid for saving $\cS_G$ in the range of $[0,10^{6}]$.\footnote{See 
	Appendix~\ref{sec:expgrid} for the construction of the exponential grid, which is based on \cite{Gouin-BonenfantTodaParetoExtrapolation}. We use the median grid point $s=10$.} 
We choose the convergence criterion such that policy iteration stops when the maximum relative change from the previous iteration satisfies
\begin{equation*}
\label{eq:crit}
    \max_{g,z} \abs{
        \frac{c_n(a_g^n (z), z)}{c_{n-1} (a_g^{n-1}(z), z)} - 1} < \epsilon = 10^{-5}.
\end{equation*}
At a small scale (Figure~\ref{fig:cf_small}), the consumption functions show a concave pattern, which is consistent with \cite{CarrollKimball1996}. At a large scale (Figure~\ref{fig:cf_large}), the consumption functions look linear, which is consistent with Theorem~\ref{thm:AL_inf} and Proposition~\ref{prop:CRRA}.

\begin{figure}[!htb]
\centering
\begin{subfigure}{0.465\linewidth}
\includegraphics[width=\linewidth]{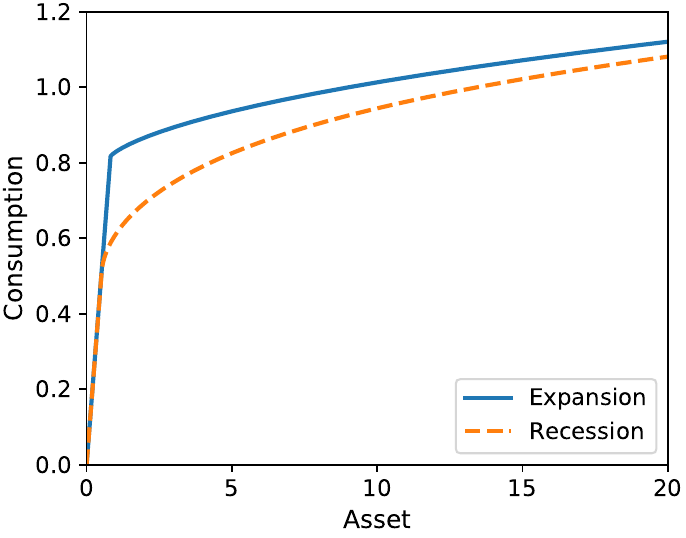}
\caption{Small scale.}\label{fig:cf_small}
\end{subfigure}
\begin{subfigure}{0.465\linewidth}
\includegraphics[width=\linewidth]{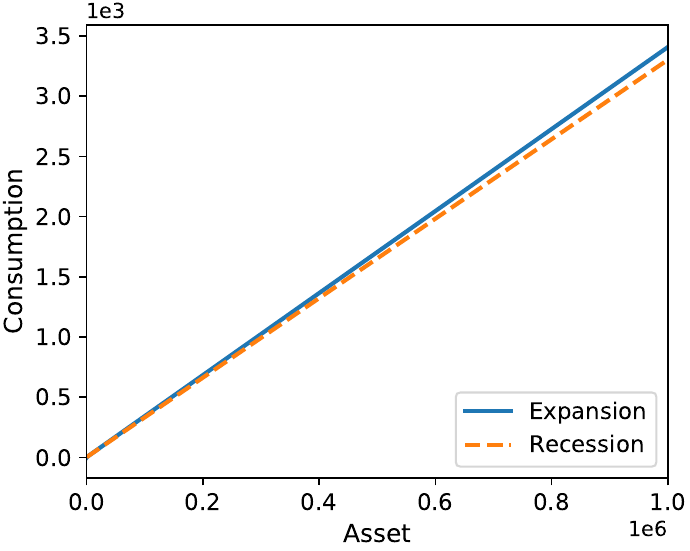}
\caption{Large scale.}\label{fig:cf_large}
\end{subfigure}
\caption{Consumption functions.}\label{fig:cf}
\end{figure}

To evaluate how fast the slope of $c(a,z)$ converges to $\bar{c}(z)$ as $a\to\infty$, we compute the marginal propensities to consume (MPCs) and their relative errors. At the grid point $a_g(z)$, the MPC and its relative error are defined as\footnote{It makes
    intuitive sense to evaluate MPC on the endogenous asset grid points. In particular, the continuity of the consumption function implies that $a_{g-1}(z) \to a_g(z)$ as $s_{g-1} \to s_g$ (see, for example, Step~\ref{item:update}\ref{item:update2} of the policy iteration algorithm).}
\begin{align*}
\mathrm{MPC}(a_g(z),z)&\coloneqq \frac{c(a_g(z),z)-c(a_{g-1}(z),z)}{a_g(z)-a_{g-1}(z)},\\
\mathrm{Error}(a_g(z),z)&\coloneqq \abs{\frac{\mathrm{MPC}(a_g(z),z)}{\bar{c}(z)}-1},
\end{align*}
respectively. Figure~\ref{fig:MPC} shows the numerical and asymptotic MPCs. Consistent with theory, the MPCs appear to converge to the theoretical values.\footnote{Although not visible from Figure~\ref{fig:MPC}, it is clear from Figure~\ref{fig:cf_large} that each state has its own limit: we have $(\bar{c}(1),\bar{c}(2))=10^{-3}\times (3.4049,3.2991)$.} Figure~\ref{fig:MPC_error} shows that the relative errors are quite small beyond $a=10^5$ (around $0.01\%$).

\begin{figure}[!htb]
\centering
\begin{subfigure}{0.465\linewidth}
\includegraphics[width=\linewidth]{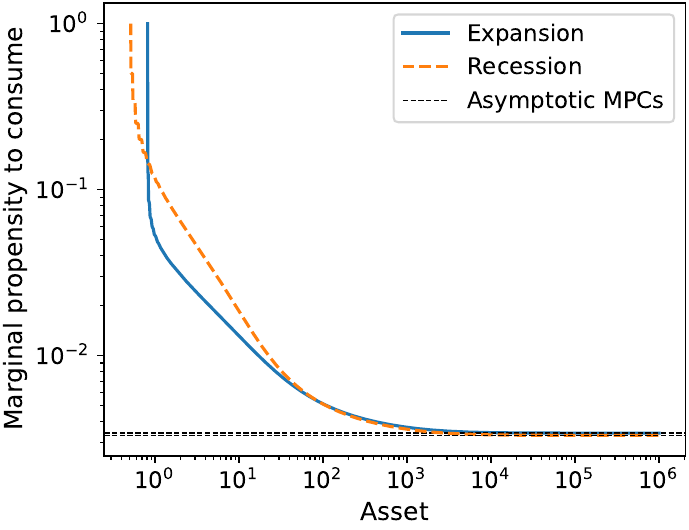}
\caption{MPCs.}\label{fig:MPC}
\end{subfigure}
\begin{subfigure}{0.465\linewidth}
\includegraphics[width=\linewidth]{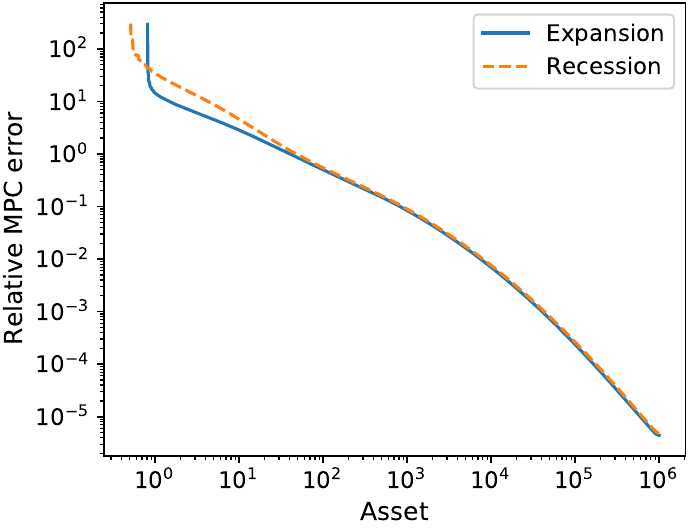}
\caption{Relative errors.}\label{fig:MPC_error}
\end{subfigure}
\caption{Marginal propensities to consume (MPCs) and relative errors.}
\end{figure}

\paragraph{The impact of maximum saving grid}

As discussed above, to compute the consumption function numerically, the state space has to be finite. Therefore, it is important to know how to set up the grid points effectively in practice. To answer this question, we study how the choice of the maximum grid point for saving affects the accuracy of the calculated consumption functions.

Our experiment is as follows: First, we treat the consumption function $c(a,z)$ in the previous section as the true consumption function, because it is calculated on a relatively large saving space $\cS_G = \set{s_g}_{g=1}^G$ with $0 = s_1 < \dots < s_G = 10^6$ and a fine exponential grid of $G = 1000$ points. We then truncate the grid points for saving to $\cS_G(\bar{s}) \coloneqq \cS_G\cap [0,\bar{s}]$, where $\bar{s}$ is the truncation point. Once this is done, for each $\bar{s}$, we compute the consumption function on the truncated grid $\cS_G(\bar{s})$, which we denote as $c(a, z; \bar{s})$, and calculate its error relative to the true consumption function via
\begin{equation*}
    \text{Error}(\bar{s}) = \max_{(a,z) \in \cA_G \times \ZZ} \abs{
        \frac{c(a,z;\bar{s})}{c(a,z)} - 1},
\end{equation*} 
where with a slight abuse of notation, we use $\cA_G \coloneqq \set{a_g(z)}_{g=1}^G{_{z=1}^Z}$ to denote the endogenous asset grid points calculated from the true consumption function $c(a,z)$. The relative error as a function of the truncation point $\bar{s}$ is displayed in Figure~\ref{fig:gmax}.

\begin{figure}[!htb]
	\centering
	\includegraphics[width=0.6\linewidth]{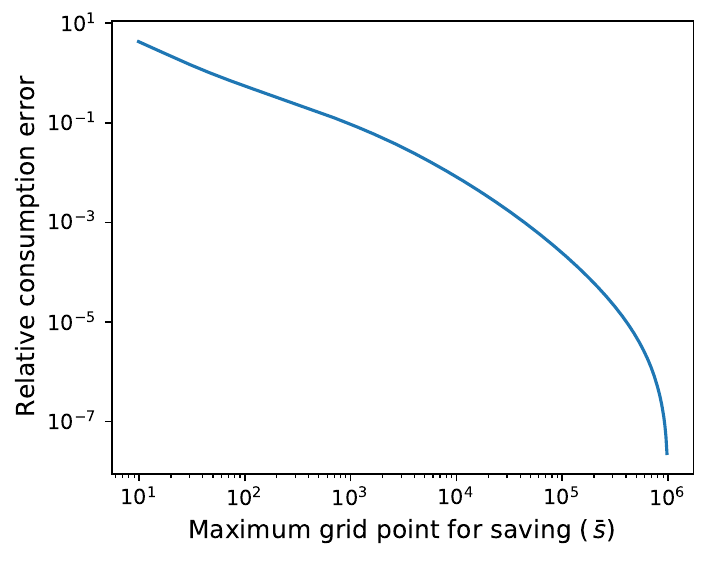}
	\caption{Relative error under different maximum saving grid point.}\label{fig:gmax}
\end{figure}

As can be seen from Figure~\ref{fig:gmax}, the error of $c(a,z;\bar{s})$ relative to $c(a,z)$ reduces greatly as the maximum grid point for saving $\bar{s}$ gets larger and becomes reasonably small for $\bar{s} > 10^4$ (below $1\%$). Intuitively, smaller saving spaces typically imply that the MPCs at the boundary endogenous asset grid points are further away from their theoretical asymptotes. Therefore, extrapolating consumption outside of the truncated space would result in larger errors. In particular, using small grids such as $\bar{s}< 10^2$ results in very large errors.\footnote{For instance, \cite{HeatonLucas1997} use a 7-point grid on $[0,3]$ for wealth to solve for consumption functions, which are likely to be inaccurate.} %The most straightforward remedy would be expanding the state space and using finer grid points, although doing that adds up the computational burden.

\paragraph{Computational efficiency}

To evaluate the computational efficiency of the policy iteration algorithm, we now solve for the consumption functions with various specifications for the number of grid points and initial condition. Same as before, we terminate the policy iteration algorithm at precision $\epsilon = 10^{-5}$. The number of grid points is $G\in \set{50,100,1000}$. To see how the initial guess affects the convergence speed, instead of \eqref{eq:c0}, we set
%\begin{equation}
%c_0(a,z)\coloneqq \min\set{a,(\alpha+(1-\alpha)\bar{c}(z))a + (1-\bar{c}(z))\bar{a}(z) },\label{eq:c0alpha}
%\end{equation}
\begin{equation}
c_0(a,z;\alpha)\coloneqq \min\set{a,\bar{c}(z;\alpha)a + (1-\bar{c}(z;\alpha))\bar{a}(z)},\label{eq:c0alpha}
\end{equation}
where $\bar{c}(z;\alpha)\coloneqq \alpha+(1-\alpha)\bar{c}(z)$ and $\alpha\in \set{0, 0.001, 0.01, 0.1, 0.2, 0.5, 1}$. For instance, setting $\alpha=1$ amounts to using $c_0(a,z)\equiv a$, while setting $\alpha=0$ amounts to using \eqref{eq:c0}. A low value of $\alpha$ implies that we choose an initial guess that has an asymptotic slope closer to the true solution.

Table~\ref{t:converge} shows the number of iterations and computing time (in seconds) required for convergence for each specification. To calculate these statistics, in each case we repeat the same solution process $50$ times and then take the average. We see that that using a theoretically motivated initial guess \eqref{eq:c0} instead of $c_0(a,z)\equiv a$ speeds up the algorithm by about $1.32$ to $1.83$ times. The enhanced computational efficiency is largely because the initial guess \eqref{eq:c0} stays closer to the true consumption function compared with $c_0(a,z) \equiv a$, in which case policy iteration converges within fewer steps. Furthermore, because the policy iteration algorithm avoids costly root-finding, the computing time is relatively insensitive to the number of grid points.

\begin{table}[!htb]
	\centering
	\caption{Speed of convergence of policy iteration.}\label{t:converge}
	\vspace{-0.25cm}
	\begin{tabular}{lrrrrrr}
		\toprule
		$\alpha$ & Iterations & Time & Iterations & Time & Iterations & Time \\
		\midrule
		& \multicolumn{2}{c}{$G=50$} & \multicolumn{2}{c}{$G=100$} & \multicolumn{2}{c}{$G=\text{1,000}$}\\
		\cmidrule(lr){2-3}
		\cmidrule(lr){4-5}
		\cmidrule(lr){6-7}
		1    	& 1,716 & 0.66 & 1,714 & 0.64 & 1,712 &   2.86 \\
		0.5     & 1,715 & 0.58 & 1,713 & 0.62 & 1,711 &   2.96 \\
		0.2     & 1,712 & 0.57 & 1,710 & 0.62 & 1,708 &   2.87 \\
		0.1    	& 1,707 & 0.85 & 1,705 & 0.62 & 1,703 &   2.91 \\
		0.01   	& 1,630 & 0.53 & 1,628 & 0.60 & 1,627 &   2.77 \\
		0.001  	& 1,278 & 0.44 & 1,276 & 0.49 & 1,275 &   2.14 \\
		0 		& 958   & 0.36 & 1,100 & 0.39 & 1,286 &   2.16 \\
		\bottomrule
	\end{tabular}
	\caption*{\footnotesize Note: the table shows the number of iterations and computing time (in seconds) required for convergence. Here $\alpha$ is the parameter in \eqref{eq:c0alpha} and $G$ is the number of grid points. For each specification, the statistics are calculated by averaging $50$ repeated experiments.}
\end{table}

Seen from Table~\ref{t:converge}, it is fair to predict that policy iteration tends to be more time-consuming as $\alpha$ becomes larger. Figure~\ref{fig:iter_time} plots the steps and time taken for policy iteration to converge over a fine grid for $\alpha$ and inspects this conjecture. In particular, we fix the number of grid points for saving at $G=1000$ and use an exponential grid for $\alpha$ in the range of $[10^{-5}, 1]$ with $50$ points. Same to Table~\ref{t:converge}, we repeat the solution process $50$ times for each specification and calculate the mean computing time. Figure~\ref{fig:iter} shows that the step for policy iteration to converge is strictly increasing in $\alpha$ unless $\alpha$ takes very small values (lower than $10^{-3}$), while Figure~\ref{fig:time} reveals a clear increasing trend in computing time as $\alpha$ gets larger. Intuitively, since a higher $\alpha$ tends to shift the initial guess further away from the true consumption function, policy iteration will take more steps to converge and the problem will be more computationally intensive.

\begin{figure}[!htb]
	\centering
	\begin{subfigure}{0.465\linewidth}
		\includegraphics[width=\linewidth]{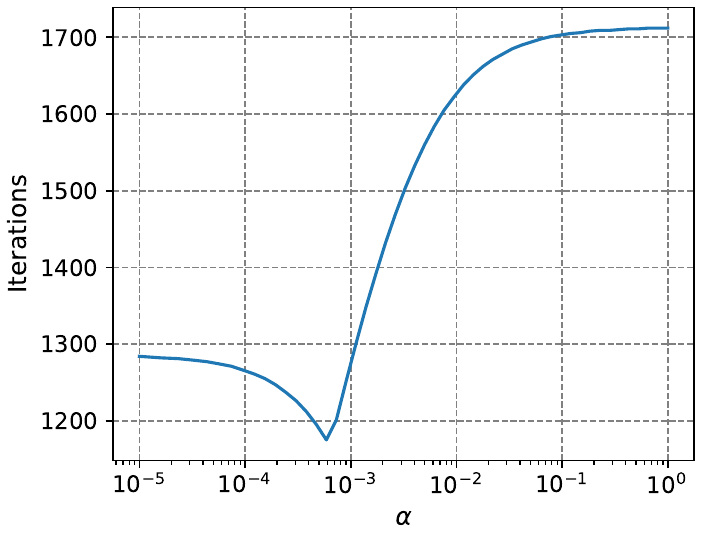}
		\caption{Number of iterations.}\label{fig:iter}
	\end{subfigure}
	\begin{subfigure}{0.465\linewidth}
		\includegraphics[width=\linewidth]{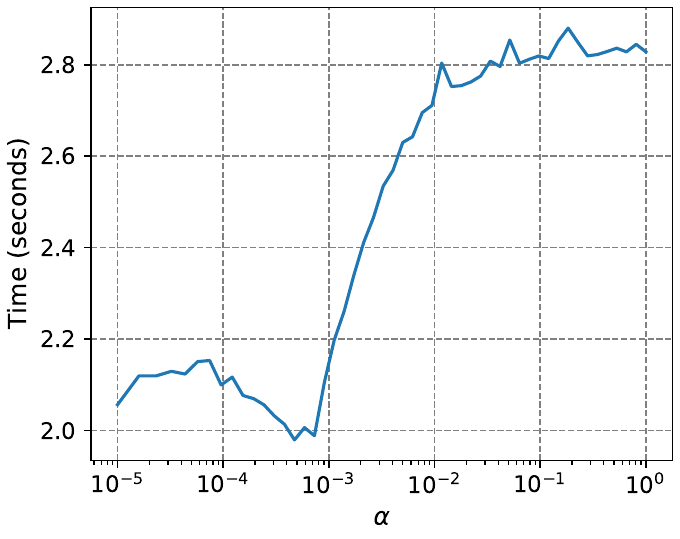}
		\caption{Time in seconds.}\label{fig:time}
	\end{subfigure}
	\caption{Speed of convergence under different $\alpha$'s.}\label{fig:iter_time}
\end{figure}

\section{Concluding remarks}\label{sec:conclu}

In this paper, we have systematically studied the asymptotic behavior of the consumption function when the marginal utility is regularly varying. We have shown that the optimal consumption function of the finite horizon optimal savings problem is always asymptotically linear in wealth. For infinite horizon problem, we have derived an asymptotic upper bound for MPC and shown that, whenever the spectral radius condition $r(K(1-\gamma)) < 1$ holds and the limit infimum of MPC as asset tends to infinity is positive, the consumption function is asymptotically linear in wealth and the asymptotic MPC coincides with the upper bound we establish. Furthermore, we have established a list of sufficient conditions for asymptotic linearity based on bounded relative risk aversion or asymptotic constant relative risk aversion.

Our results build a theoretical foundation for linearly extrapolating consumption outside the grid points when solving the optimal savings problem numerically. This in turn allows us to construct good initial guesses for policy iteration and solve the problem efficiently and accurately. Our numerical experiments have demonstrated that working with the initial guess defined in \eqref{eq:c0} speeds up computation obviously and that working with a large saving/asset space created by exponential grid is necessary for improving solution accuracy.

\appendix
\section{Proofs}\label{sec:proof}

We need the following result to prove Proposition \ref{prop:RV}.

\begin{thm}[Representation Theorem]\label{thm:RT}
Let $\ell:(0,\infty)\to (0,\infty)$ be measurable. Then $\ell$ is slowly varying if and only if it may be written in the form
\begin{equation}
\ell(x)=\exp\left(\eta(x)+\int_a^x \frac{\varepsilon(t)}{t}\diff t\right) \label{eq:RT}
\end{equation}
for $x\ge a$ for some $a>0$, where $\eta,\varepsilon$ are measurable functions such that $\eta(x)\to \eta\in \R$ and $\varepsilon(x)\to 0$ as $x\to\infty$.
\end{thm}

\begin{proof}
See \citet[Theorem 1.3.1]{BinghamGoldieTeugels1987}.
\end{proof}

\begin{proof}[Proof of Proposition \ref{prop:RV}] 
\begin{step}
$\aCRRA(\gamma)\subset \RV(-\gamma)$.
\end{step}
If $u$ is asymptotically CRRA with coefficient $\gamma$, by definition we can write 
$$\frac{cu''(c)}{u'(c)}=-\gamma+\varepsilon(c),$$
where $\varepsilon(c)\to 0$ as $c\to\infty$. Dividing both sides by $c>1$ and integrating on $[1,c]$, we obtain
\begin{align*}
&\log \frac{u'(c)}{u'(1)}=-\gamma\log c+\int_1^c \frac{\varepsilon(t)}{t}\diff t \\
\iff &u'(c)=c^{-\gamma}\exp\left(\log u'(1)+\int_1^c \frac{\varepsilon(t)}{t}\diff t\right).
\end{align*}
Therefore by Theorem \ref{thm:RT}, $u'$ is regularly varying with index $-\gamma$ by setting $u'(c)=c^{-\gamma}\ell(c)$ with $\ell$ defined by \eqref{eq:RT} with $\eta(x)\equiv \log u'(1)$ and $a=1$.

\begin{step}
$\RV(-\gamma) \not\subset \aCRRA(\gamma)$.
\end{step}
Consider the marginal utility function $u'(c)=c^{-\gamma}\ell(c)$ for $\ell$ in \eqref{eq:RT}, where
$$\eta(x)=\delta\int_0^x\frac{\sin t}{t}\diff t$$
for some $\delta\in (0,\gamma)$ and $\varepsilon(x)\equiv 0$. Noting that $\int_0^\infty \sin t/t\diff t=\pi/2$,\footnote{See \cite{Hardy1909} for an interesting discussion of this integral.} we have $\eta(x)\to \pi\delta/2$ as $x\to\infty$, so by Theorem \ref{thm:RT}, $\ell$ is slowly varying and $u\in \RV(-\gamma)$. However, log differentiating $u'$, we obtain
$$\frac{u''(c)}{u'(c)}=-\frac{\gamma}{c}+\delta\frac{\sin c}{c}\iff -\frac{cu''(c)}{u'(c)}=\gamma-\delta\sin c>0,$$
so $u''<0$ always but $-cu''(c)/u'(c)$ does not converge as $c\to\infty$. Therefore $u\notin \aCRRA(\gamma)$.

\begin{step}
$\RV(-\gamma)\subset \aE(-\gamma)$.
\end{step}
If $u'$ is regularly varying with index $-\gamma$, by Theorem \ref{thm:RT} we can write $u'(c)=c^{-\gamma}\ell(c)$ with $\ell$ as in \eqref{eq:RT}. Take any $\delta>0$. Since $\varepsilon(c)\to 0$ as $c\to\infty$, we can take $\bar{c}>\max\set{a,1}$ such that $\abs{\varepsilon(c)}\le \delta$ for $c\ge \bar{c}$. Then
\begin{align*}
\abs{\log \ell(c)}&=\abs{\eta(c)+\int_a^c\frac{\varepsilon(t)}{t}\diff t}\le \abs{\eta(c)}+\int_a^c\frac{\abs{\varepsilon(t)}}{t}\diff t\\
&\le \abs{\eta(c)}+\int_a^{\bar{c}}\frac{\abs{\varepsilon(t)}}{t}\diff t+\int_{\bar{c}}^c\frac{\delta}{t}\diff t\\
&=\abs{\eta(c)}+\int_a^{\bar{c}}\frac{\abs{\varepsilon(t)}}{t}\diff t+\delta\log\frac{c}{\bar{c}}.
\end{align*}
Dividing both sides by $\log c>\log \bar{c}>0$ and letting $c\to\infty$, noting that $\eta(c)\to \eta$ as $c\to\infty$, we obtain
$$\limsup_{c\to\infty}\abs{\frac{\log \ell(c)}{\log c}}\le \delta.$$
Since $\delta>0$ is arbitrary, letting $\delta\downarrow 0$, we obtain
$$\lim_{c\to\infty}\frac{\log \ell(c)}{\log c}=0.$$
Therefore $\log u'(c)/\log c\to -\gamma$ because $u'(c)=c^{-\gamma}\ell(c)$, and by definition $u'$ has an asymptotic exponent $-\gamma$.

\begin{step}
$\aE(-\gamma) \not\subset \RV(-\gamma)$.
\end{step}
Consider the marginal utility function $u'(c)=c^{-\gamma}\exp(\delta \sin \log c)$, where $\delta\in (0,\gamma)$. Then
$$\frac{\log u'(c)}{\log c}=-\gamma+\delta\frac{\sin \log c}{\log c}\to -\gamma$$
as $c\to\infty$, so $u\in \aE(-\gamma)$. Furthermore, by log differentiating $u'$, we obtain
$$\frac{u''(c)}{u'(c)}=-\frac{\gamma}{c}+\delta\frac{\cos \log c}{c}\iff -\frac{cu''(c)}{u'(c)}=\gamma-\delta\cos\log c>0,$$
so $u''<0$ always. To show $u\notin \RV(-\gamma)$, it suffices to show that $\ell(c)\coloneqq \exp(\delta \sin \log c)$ is not slowly varying. Take $\lambda=\e^{\pi}$ and $c_n=\e^{\pi(n-1/2)}$. Then
\begin{align*}
\log \frac{\ell(\lambda c_n)}{\ell(c_n)}&=\delta(\sin \log(\lambda c_n)-\sin \log c_n)\\
&=\delta (\sin (\pi(n+1/2))-\sin (\pi(n-1/2)))\\
&=2(-1)^n\delta,
\end{align*}
which does not converge as $n\to\infty$. Therefore $\ell$ is not slowly varying.
\end{proof}

Before proving Theorems~\ref{thm:AL_finite} and \ref{thm:AL_inf}, we first proceed heuristically to motivate what the value of $\bar{c}(z)=\lim_{a\to\infty}c(a,z)/a$ should be if it exists. Assuming that the borrowing constraint does not bind, the Euler equation \eqref{eq:xi} implies
$$u'(\xi)=\E_z\hat{\beta}\hat{R}u'(c(\hat{R}(a-\xi)+\hat{Y},\hat{Z})),$$
where $\xi=c(a,z)$. Multiplying both sides by $a^\gamma$, setting $c(a,z)=\bar{c}(z)a$ motivated by \eqref{eq:MPC_lim}, letting $a\to\infty$, using Assumption \ref{asmp:regular} (regular variation), and interchanging expectations and limits, it must be
$$\bar{c}(z)^{-\gamma}=\E_z\hat{\beta}\hat{R}^{1-\gamma}\bar{c}(\hat{Z})^{-\gamma}(1-\bar{c}(z))^{-\gamma}.$$
Dividing both sides by $(1-\bar{c}(z))^{-\gamma}$ and setting $x(z)=\bar{c}(z)^{-\gamma}$, we obtain
\begin{equation}
x(z)=\left(1+\left(\E_z\hat{\beta}\hat{R}^{1-\gamma}x(\hat{Z})\right)^{1/\gamma}\right)^\gamma,\quad  z=1,\dots,Z.\label{eq:xzeq}
\end{equation}
Noting that
\begin{equation}
\E_z\hat{\beta}\hat{R}^{1-\gamma}x(\hat{Z})=\sum_{\hat{z}=1}^ZP_{z\hat{z}}\E_{z,\hat{z}}\hat{\beta}\hat{R}^{1-\gamma}x(\hat{z})\label{eq:Ezsum}
\end{equation}
and using the definition of $K$ in \eqref{eq:Ktheta}, we can rewrite \eqref{eq:xzeq} as $x=Fx$, where $F$ is defined in \eqref{eq:F}.

We apply policy function iteration to prove Theorems~\ref{thm:AL_finite} and \ref{thm:AL_inf}. Let $\cC$ be the space of candidate consumption functions as defined in Section \ref{sec:IF}. We further restrict the candidate space to satisfy asymptotic linearity:
\begin{equation}
\cC_1=\set{c\in \cC|(\forall z\in \ZZ) \exists \bar{c}(z)=\lim_{a\to\infty}\frac{c(a,z)}{a}\in (0,1]}.\label{eq:C1}
\end{equation}
Clearly $\cC_1$ is nonempty because $c(a,z)\equiv a$ belongs to $\cC_1$.

The following proposition shows that the operator $T$ defined in Section \ref{sec:IF} maps the candidate space $\cC_1$ into itself and also shows how the asymptotic MPCs of $c$ and $Tc$ are related.

\begin{prop}\label{prop:C1}
Suppose Assumptions~\ref{asmp:Inada}, \ref{asmp:spectral}\ref{item:Ebeta}\ref{item:EY}, \ref{asmp:regular}, and \ref{asmp:R} hold. Then $T\cC_1\subset \cC_1$. For $c\in \cC_1$, let $\bar{c}(z)=\lim_{a\to\infty}c(a,z)/a$ and $x(z)=\bar{c}(z)^{-\gamma}\in [1,\infty)$. Then
\begin{equation}
\lim_{a\to\infty}\frac{Tc(a,z)}{a}=(Fx)(z)^{-1/\gamma},\label{eq:Tcbar}
\end{equation}
where $F$ is as in \eqref{eq:F}.
\end{prop}

To prove Proposition~\ref{prop:C1}, we need the following lemma.

\begin{lem}\label{lem:MRS}
Let $f:(0,\infty)\to(0,\infty)$ be a positive measurable function such that $\lambda=\lim_{c\to\infty}f(c)\in [0,\infty)$ exists. If Assumptions \ref{asmp:Inada} and \ref{asmp:regular} hold, then
\begin{equation}
\lim_{c\to\infty} \frac{u'(f(c)c)}{u'(c)}=\lambda^{-\gamma}.\label{eq:MRSlim}
\end{equation}
\end{lem}

\begin{proof}
Suppose $\lambda>0$. Take any numbers $\ubar{\lambda},\bar{\lambda}$ such that $0<\ubar{\lambda}<\lambda<\bar{\lambda}$. Since $f(c)\to\lambda$ as $c\to\infty$, there exists $\ubar{c}>0$ such that $f(c)\in [\ubar{\lambda},\bar{\lambda}]$ for $c\ge \ubar{c}$. Since $u'$ is strictly decreasing by Assumption \ref{asmp:Inada}, it follows that $u'(\ubar{\lambda}c)\ge u'(f(c)c)\ge u'(\bar{\lambda}c)$ for $c\ge \ubar{c}$. Dividing both sides by $u'(c)$, letting $c\to\infty$, and using Assumption \ref{asmp:regular}, we obtain
$$\ubar{\lambda}^{-\gamma}=\lim_{c\to\infty}\frac{u'(\ubar{\lambda}c)}{u'(c)}\ge \limsup_{c\to\infty}\frac{u'(f(c)c)}{u'(c)}\ge \liminf_{c\to\infty}\frac{u'(f(c)c)}{u'(c)}\ge \lim_{c\to\infty}\frac{u'(\bar{\lambda} c)}{u'(c)}=\bar{\lambda}^{-\gamma}.$$
Letting $\ubar{\lambda},\bar{\lambda}\to\lambda$, we obtain \eqref{eq:MRSlim}.

If $\lambda=0$, take any $\bar{\lambda}>0$. By the same argument as above, we obtain
$$\liminf_{c\to\infty}\frac{u'(f(c)c)}{u'(c)}\ge \lim_{c\to\infty}\frac{u'(\bar{\lambda} c)}{u'(c)}=\bar{\lambda}^{-\gamma},$$
so letting $\bar{\lambda}\downarrow 0$, we obtain \eqref{eq:MRSlim} (and both sides are $\infty$).
\end{proof}

\begin{proof}[Proof of Proposition~\ref{prop:C1}]
Let $c\in \cC_1$ and $\bar{c}(z)=\lim_{a\to\infty}c(a,z)/a\in (0,1]$.

For $\alpha\in [0,1]$, define
\begin{equation}
g_c(\alpha,a,z)=\frac{u'(\alpha a)}{u'(a)}-\min\set{\max\set{\E_z \hat{\beta}\hat{R}\frac{u'(c(\hat{R}(1-\alpha)a+\hat{Y},\hat{Z}))}{u'(a)},1},\frac{u'(0)}{u'(a)}}.\label{eq:gc}
\end{equation}
By Assumption \ref{asmp:Inada}, if $u'(0)<\infty$, then $g_c$ is continuous and strictly decreasing in $\alpha\in [0,1]$ with $g_c(0,a,z)\ge 0$ and $g_c(1,a,z)\le 0$. If $u'(0)=\infty$, then $g_c$ is continuous and strictly decreasing in $\alpha\in (0,1]$ with $g_c(0,a,z)=\infty$ and $g_c(1,a,z)\le 0$. In either case, by the intermediate value theorem, for each $(a,z)$, there exists a unique $\alpha\in [0,1]$ such that $g_c(\alpha,a,z)=0$. By the definition of $T$, we have $g_c(\xi/a,a,z)=0$, where $\xi=Tc(a,z)$. Therefore $\alpha=Tc(a,z)/a$.

If $\hat{\beta}\hat{R}=0$ almost surely conditional on $Z=z$, then \eqref{eq:gc} becomes
$$g_c(\alpha,a,z)=\frac{u'(\alpha a)}{u'(a)}-1.$$
Since $\alpha=Tc(a,z)/a$ solves $g_c(\alpha,a,z)=0$, it must be $Tc(a,z)/a=\alpha=1$. Therefore in particular $\lim_{a\to\infty}Tc(a,z)/a=1$ exists and \eqref{eq:Tcbar} is trivial. Below, assume $\hat{\beta}\hat{R}>0$ with positive probability conditional on $Z=z$.

Take any accumulation point $\alpha$ of $Tc(a,z)/a\in [0,1]$ as $a\to\infty$, which always exists because $0\le Tc(a,z)/a\le 1$. Then we can take an increasing sequence $\set{a_n}$ such that $\alpha=\lim_{n\to\infty} Tc(a_n,z)/a_n$. Define $\alpha_n=Tc(a_n,z)/a_n\in [0,1]$ and
\begin{equation}
\lambda_n=\frac{c(\hat{R}(1-\alpha_n)a_n+\hat{Y},\hat{Z})}{a_n}\ge 0.\label{eq:lambdan}
\end{equation}
By the definitions of $\alpha_n$ and $\lambda_n$, we have
\begin{align}
&0=g_c(\alpha_n,a_n,z)=\frac{u'(\alpha_na_n)}{u'(a_n)}-\min\set{\max\set{\E_z\hat{\beta}\hat{R}\frac{u'(\lambda_n a_n)}{u'(a_n)},1},\frac{u'(0)}{u'(a_n)}}\notag \\
&\implies \frac{u'(\alpha_na_n)}{u'(a_n)}=\min\set{\max\set{\E_z\hat{\beta}\hat{R}\frac{u'(\lambda_n a_n)}{u'(a_n)},1},\frac{u'(0)}{u'(a_n)}}.\label{eq:eulern}
\end{align}

\begin{step}
For $\lambda_n$ in \eqref{eq:lambdan}, we have
\begin{equation}
\lim_{n\to\infty}\lambda_n=\bar{c}(\hat{Z})\hat{R}(1-\alpha).\label{eq:lambdalim}
\end{equation}
\end{step}
To see this, if $\alpha<1$ and $\hat{R}>0$, then since $\hat{R}(1-\alpha_n)a_n\to \hat{R}(1-\alpha)\cdot \infty=\infty$, by the definition of $\bar{c}$ we have
$$\lambda_n=\frac{c(\hat{R}(1-\alpha_n)a_n+\hat{Y},\hat{Z})}{\hat{R}(1-\alpha_n)a_n+\hat{Y}}\left(\hat{R}(1-\alpha_n)+\frac{\hat{Y}}{a_n}\right)\to \bar{c}(\hat{Z})\hat{R}(1-\alpha),$$
which is \eqref{eq:lambdalim}. If $\alpha=1$ or $\hat{R}=0$, then since $c(a,z)\le a$, we have
\begin{align*}
\lambda_n&=\frac{c(\hat{R}(1-\alpha_n)a_n+\hat{Y},\hat{Z})}{\hat{R}(1-\alpha_n)a_n+\hat{Y}}\left(\hat{R}(1-\alpha_n)+\frac{\hat{Y}}{a_n}\right)\\
&\le \hat{R}(1-\alpha_n)+\frac{\hat{Y}}{a_n}\to \hat{R}(1-\alpha)=0,
\end{align*}
so again \eqref{eq:lambdalim} holds.

\begin{step}
For any accumulation point $\alpha$ of $Tc(a,z)/a\in [0,1]$ as $a\to\infty$, we have $\alpha<1$.
\end{step}
Letting $n\to\infty$ in \eqref{eq:eulern}, applying Lemma \ref{lem:MRS}, and using the continuity of max and min operators, we obtain
\begin{align}
\alpha^{-\gamma}&=\lim_{n\to\infty}\frac{u'(\alpha_na_n)}{u'(a_n)} \notag \\
&=\lim_{n\to\infty}\min\set{\max\set{\E_z\hat{\beta}\hat{R}\frac{u'(\lambda_n a_n)}{u'(a_n)},1},\frac{u'(0)}{u'(a_n)}} \notag \\
&=\min\set{\max\set{\lim_{n\to\infty}\E_z\hat{\beta}\hat{R}\frac{u'(\lambda_n a_n)}{u'(a_n)},1},\lim_{n\to\infty}\frac{u'(0)}{u'(a_n)}} \notag \\
&=\max\set{\lim_{n\to\infty}\E_z\hat{\beta}\hat{R}\frac{u'(\lambda_n a_n)}{u'(a_n)},1}, \label{eq:eulerlim}
\end{align}
where the last equation uses $u'(0)/u'(a_n)\to u'(0)/u'(\infty)=\infty$ (because $u'(\infty)=0$). Applying Fatou's Lemma, \eqref{eq:lambdalim}, and Lemma \ref{lem:MRS}, it follows from \eqref{eq:eulerlim} that
\begin{align}
\alpha^{-\gamma}&\ge \E_z\hat{\beta}\hat{R}\lim_{n\to\infty}\frac{u'(\lambda_n a_n)}{u'(a_n)} \notag \\
&=\E_z \hat{\beta}\hat{R}[\bar{c}(\hat{Z})\hat{R}(1-\alpha)]^{-\gamma}\notag \\
&=\E_z \hat{\beta}\hat{R}^{1-\gamma}[\bar{c}(\hat{Z})(1-\alpha)]^{-\gamma}.\label{eq:alphaineq}
\end{align}
Since by assumption $\hat{\beta}\hat{R}>0$ with positive probability conditional on $Z=z$ and $\bar{c}(z)>0$ for all $z$ (because $c\in \cC_1$; see \eqref{eq:C1}), it follows that $\E_z \hat{\beta}\hat{R}^{1-\gamma}\bar{c}(\hat{Z})^{-\gamma}\in (0,\infty)$. Therefore solving the inequality \eqref{eq:alphaineq}, we obtain
\begin{equation*}
\alpha \le \frac{1}{1+\left(\E_z\hat{\beta}\hat{R}^{1-\gamma}\bar{c}(\hat{Z})^{-\gamma}\right)^{1/\gamma}}<1.
\end{equation*}

\begin{step}
The limit $n\to\infty$ and the expectation $\E_z$ can be interchanged in \eqref{eq:eulerlim}.
\end{step}

Note that
\begin{equation}
\E_z\hat{\beta}\hat{R}\frac{u'(\lambda_n a_n)}{u'(a_n)}=\sum_{\hat{z}=1}^ZP_{z\hat{z}}\E_{z,\hat{z}}\hat{\beta}\hat{R}\frac{u'(\lambda_n a_n)}{u'(a_n)}.\label{eq:Ez}
\end{equation}
When computing the expectation \eqref{eq:Ez}, we can divide it into the events $\hat{R}=0$ and $\hat{R}>0$. When $\hat{R}=0$, the integrand is zero. When $\hat{R}>0$, by Assumption \ref{asmp:R}, we have $\hat{R}\ge \delta>0$. By the definition of $\lambda_n$ and the monotonicity of consumption functions established in \cite{MaStachurskiToda2020JET}, it follows from the definition of $\lambda_n$ in \eqref{eq:lambdan} that
$$\lambda_n\ge \frac{c(\delta(1-\alpha_n)a_n,\hat{Z})}{a_n}\to \bar{c}(\hat{Z})\delta(1-\alpha).$$
Since $\bar{c}(z)>0$ for all $z$, $\alpha<1$, and $\ZZ$ is a finite set, for any 
$$\ubar{\lambda}\in \left(0,\min_{z\in \ZZ}\bar{c}(z)\delta(1-\alpha)\right),$$
there exists $N$ such that $\lambda_n\ge \ubar{\lambda}$ for all $n\ge N$ and $\hat{Z}\in \ZZ$. Then by Assumptions \ref{asmp:Inada} and \ref{asmp:regular}, for $n\ge N$ we have
$$\frac{u'(\lambda_na_n)}{u'(a_n)}\le \frac{u'(\ubar{\lambda} a_n)}{u'(a_n)}\to \ubar{\lambda}^{-\gamma}$$
as $n\to\infty$. Therefore for any $M\in (\ubar{\lambda}^{-\gamma},\infty)$, we have
$$\frac{u'(\lambda_na_n)}{u'(a_n)}\le M<\infty$$
for large enough $n$. Since by Assumption \ref{asmp:spectral} we have $\E_{z,\hat{z}}\hat{\beta}\hat{R}<\infty$ whenever $P_{z\hat{z}}>0$, it follows from \eqref{eq:eulerlim}, the dominated convergence theorem, and Lemma \ref{lem:MRS} that
\begin{align}
\alpha^{-\gamma}&=\max\set{\lim_{n\to\infty}\E_z\hat{\beta}\hat{R}\frac{u'(\lambda_n a_n)}{u'(a_n)},1} \notag \\
&=\max\set{\E_z\hat{\beta}\hat{R}\lim_{n\to\infty}\frac{u'(\lambda_n a_n)}{u'(a_n)},1} \notag \\
&=\max\set{\E_z \hat{\beta}\hat{R}^{1-\gamma}[\bar{c}(\hat{Z})(1-\alpha)]^{-\gamma},1}.\label{eq:alpha}
\end{align}

\begin{step}
The limit \eqref{eq:Tcbar} holds.
\end{step}

Since the left-hand side of \eqref{eq:alpha} is strictly decreasing in $\alpha$ and the right-hand side is weakly increasing in $\alpha$, the number $\alpha$ that solves \eqref{eq:alpha} is unique. Since $\alpha$ is any accumulation point of $Tc(a,z)/a\in [0,1]$ as $a\to\infty$, it follows that $\lim_{a\to\infty}Tc(a,z)/a$ exists. Therefore it only remains to show that the limit $\alpha$ equals the right-hand side of \eqref{eq:Tcbar}.

If $\alpha=0$, then \eqref{eq:alpha} implies
$$\infty=\max\set{\E_z \hat{\beta}\hat{R}^{1-\gamma}\bar{c}(\hat{Z})^{-\gamma},1}<\infty,$$
which is a contradiction. Since $\alpha<1$, \eqref{eq:alpha} implies
\begin{align*}
&\alpha^{-\gamma}=\E_z \hat{\beta}\hat{R}^{1-\gamma}[\bar{c}(\hat{Z})(1-\alpha)]^{-\gamma}\\
\iff &\alpha=\frac{1}{1+\left(\E_z\hat{\beta}\hat{R}^{1-\gamma}\bar{c}(\hat{Z})^{-\gamma}\right)^{1/\gamma}}=(Fx)(z)^{-1/\gamma}.\qedhere
\end{align*}
\end{proof}

With all the above preparations, we can prove Theorems~\ref{thm:AL_finite} and \ref{thm:AL_inf}.

\begin{proof}[Proof of Theorem~\ref{thm:AL_finite}]
By Theorem~\ref{thm:exist}, we have $c_n=T^nc_0$, where $c_0(a,z)\coloneqq a$. The limit \eqref{eq:cnbar} follows from applying Proposition~\ref{prop:C1} $n$ times.
\end{proof}

\setcounter{case}{0}
\begin{proof}[Proof of Theorem \ref{thm:AL_inf}]
Lemma B.4 of \cite{MaStachurskiToda2020JET} shows that $T:\cC\to \cC$ is order preserving, that is, $c_1\le c_2$ implies $Tc_1\le Tc_2$. Define the sequence $\set{c_n}\subset \cC$ by $c_0(a,z)=a$ and $c_n=Tc_{n-1}$ for all $n\ge 1$. Since $c_0(a,z)/a=1$, in particular we have $c_0\in \cC_1$, where $\cC_1$ is as in \eqref{eq:C1}. Therefore by Proposition \ref{prop:C1}, we have $c_n\in \cC_1$ for all $n$, so $\bar{c}_n(z)=\lim_{a\to\infty}c_n(a,z)/a\in (0,1]$ exists. Since $Tc(a,z)\le a$ for any $c\in \cC$, in particular $c_1(a,z)=Tc_0(a,z)\le a=c_0(a,z)$, so by induction $c_{n+1}\le c_n$ for all $n$. Define $c(a,z)=\lim_{n\to\infty}c_n(a,z)$, which exists because $\set{c_n}$ is monotonically decreasing and $c_n\ge 0$. Then by Theorem 2.2 of \cite{MaStachurskiToda2020JET}, this $c$ is the unique fixed point of $T$ and also the unique solution to the optimal savings problem \eqref{eq:IF}. Since $0\le c\le c_n$ point-wise, by Proposition \ref{prop:C1} we have
\begin{equation}
0\le \limsup_{a\to\infty}\frac{c(a,z)}{a}\le \limsup_{a\to\infty}\frac{c_n(a,z)}{a}=x_n(z)^{-1/\gamma},\label{eq:barcub}
\end{equation}
where $\set{x_n}_{n=1}^\infty \subset \R_+^Z$ is as in Theorem~\ref{thm:AL_finite}.

\begin{case}[$r(K(1-\gamma))\ge 1$ and $K(1-\gamma)$ is irreducible]
By Proposition 14 of \cite{MaToda2021JET}, we have $x_n(z)\to \infty$ for all $z\in \ZZ$. Letting $n\to\infty$ in \eqref{eq:barcub}, we obtain
$$\lim_{a\to\infty}\frac{c(a,z)}{a}=0.$$
\end{case}
\begin{case}[$r(K(1-\gamma))<1$]
By Proposition 14 of \cite{MaToda2021JET}, we have $x_n(z)\to x^*(z)$, where $x^*$ is the unique fixed point of $F$ in \eqref{eq:F}. Letting $n\to\infty$ in \eqref{eq:barcub}, we obtain
$$0\le \liminf_{a\to\infty}\frac{c(a,z)}{a}\le \limsup_{a\to\infty}\frac{c(a,z)}{a}\le x^*(z)^{-1/\gamma},$$
which is \eqref{eq:MPC_bound}. If $\liminf_{a\to\infty}c(a,z)/a>0$ for all $z$, we can take $\epsilon(z)>0$ such that $\liminf_{a\to\infty}c(a,z)/a>\epsilon(z)>0$. Therefore we can take $\bar{a}(z)>0$ such that $c(a,z)>\epsilon(z)a$ for $a\ge \bar{a}(z)$. Define $\ubar{a}(z)=\inf \set{a|c(a,z)\ge \epsilon(z)\bar{a}(z)}$. Since $c(a,z)$ is increasing and continuous in $a$ and $c(a,z)\le a$, we have $0<\ubar{a}(z)<\bar{a}(z)$. Furthermore, define
$$c_0(a,z)=\begin{cases}
c(a,z), & (0<a\le \ubar{a}(z))\\
\epsilon(z)\bar{a}(z), & (\ubar{a}(z)<a\le \bar{a}(z))\\
\epsilon(z)a, & (a>\bar{a}(z))
\end{cases}$$
as in Figure \ref{fig:c0}.

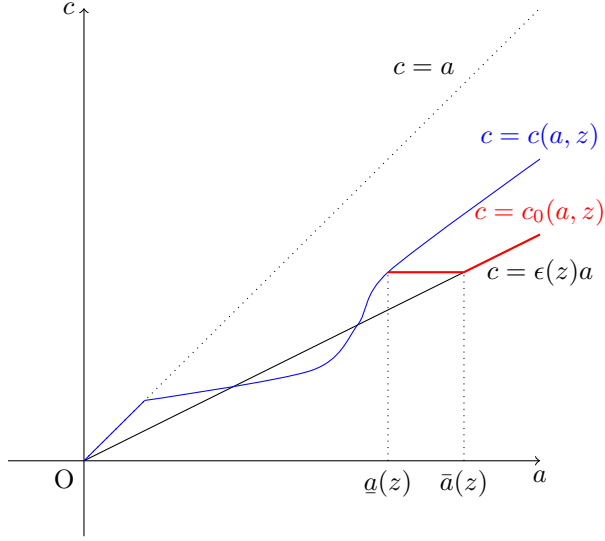
\begin{figure}[!htb]
\centering
\begin{tikzpicture}

% draw axes
\draw [->] (-1,0)--(6,0);
\draw [->] (0,-1)--(0,6);
\draw (0,0) node[anchor = north east] {$\mathrm{O}$};
\draw (6,0) node[anchor = north] {$a$};
\draw (0,6) node[anchor = east] {$c$};

% draw 45 degree line
\draw [dotted] (0,0)--(6,6);
\draw (5,5) node[anchor = south east] {$c=a$};

% draw linear lower bound
\draw (0,0)--(6,3);
\draw (6,2.5) node {$c=\epsilon(z)a$};

% draw consumption function
\draw [blue] (0,0)--(0.8,0.8);
\draw [blue] plot [smooth] coordinates {(0.8,0.8) (3,1.2) (3.6,1.8) (4,2.5) (6,4)};
%\draw [blue] (5,3.3)--(6,4);
\draw [blue] (6,4) node[anchor = south] {$c=c(a,z)$};

\draw [dotted] (4,0)--(4,2.5);
\draw [dotted] (5,0)--(5,2.5);
\draw (4,0) node[anchor = north] {$\ubar{a}(z)$};
\draw (5,0) node[anchor = north] {$\bar{a}(z)$};

\draw [red,thick] (4,2.5)--(5,2.5)--(6,3);
\draw [red] (6,3) node[anchor = south] {$c=c_0(a,z)$};

\end{tikzpicture}
\caption{Definition of $c_0(a,z)$.}\label{fig:c0}
\end{figure}

By definition, $c_0\le c$ point-wise. Let us show that $c_0\in \cC_1$. Since $c(a,z)$ is increasing and continuous in $a$, so is $c_0(a,z)$. Clearly $0<c_0(a,z)\le c(a,z)\le a$. Because $c\in \cC$ and $c_0(a,z)=c(a,z)$ for $a\le \ubar{a}(z)$, we have
$$\sup_{a\le \ubar{a}(z)}\abs{u'(c_0(a,z))-u'(a)}=\sup_{a\le \ubar{a}(z)}\abs{u'(c(a,z))-u'(a)}<\infty.$$
Since $u'>0$, $c_0(a,z)\le a$, and $u'$ is decreasing, we have
$$\sup_{a>\ubar{a}(z)}\abs{u'(c_0(a,z))-u'(a)}\le \sup_{a>\ubar{a}(z)}u'(c_0(a,z))=u'(\epsilon(z)\bar{a}(z))<\infty.$$
Since $\ZZ$ is a finite set, we have
$$\sup_{(a,z)\in (0,\infty)\times \ZZ}\abs{u'(c_0(a,z))-u'(a)}<\infty,$$
so $c_0\in \cC$. Since $c_0(a,z)=\epsilon(z)a$ for $a>\bar{a}(z)$, we have $c_0(a,z)/a\to \epsilon(z)\in (0,1]$ as $a\to\infty$, so $c_0\in \cC_1$.

Since $c\ge c_0$ and $c_0\in\cC_1$, by iteration $c\ge c_n\coloneqq T^nc_0$ for all $n$. By Proposition \ref{prop:C1}, we have
$$\liminf_{a\to\infty}\frac{c(a,z)}{a}\ge \lim_{a\to\infty}\frac{c_n(a,z)}{a}=x_n(z)^{-1/\gamma},$$
where $\set{x_n}\subset \R_+^Z$ is defined by $x_0(z)=\epsilon(z)^{-\gamma}<\infty$ and iterating $x_n=Fx_{n-1}$. By Proposition 14 of \cite{MaToda2021JET}, we have $x_n\to x^*$ as $n\to\infty$, so
\begin{equation*}
x^*(z)^{-1/\gamma}\ge \limsup_{a\to\infty}\frac{c(a,z)}{a}\ge \liminf_{a\to\infty}\frac{c(a,z)}{a}\ge x^*(z)^{-1/\gamma}.
\end{equation*}
Therefore $\lim_{a\to\infty}c(a,z)/a=x^*(z)^{-1/\gamma}$. \qedhere
\end{case}
\end{proof}

\begin{proof}[Proof of Theorem~\ref{thm:suff_cond}]
Restrict the candidate space to
\begin{equation}
\cC_2=\set{c\in \cC| c(a,z)\ge \epsilon(z)a~\text{for all}~a>A~\text{and}~z\in \ZZ}. \label{eq:C2}
\end{equation}
Clearly $\cC_2$ is nonempty because $c(a,z)\equiv a$ belongs to $\cC_2$. Let us show that $T\cC_2\subset \cC_2$. Suppose to the contrary that $T\cC_2\not\subset \cC_2$. Then there exists $c\in \cC_2$ such that for some $a>A$ and $z \in \ZZ$, we have $\xi\coloneqq Tc(a,z)<\epsilon(z)a$. 

Since $u'$ is strictly decreasing and $\epsilon(z)\in (0,1)$, it follows from \eqref{eq:xi} that
\begin{align}
u'(a) &< u'(\epsilon(z)a)<u'(\xi) \notag \\
&=\min\set{\max\set{\E_z \hat{\beta}\hat{R}u'(c(\hat{R}(a-\xi)+\hat{Y},\hat{Z})),u'(a)},u'(0)}.\label{eq:eulerub}
\end{align}
If $u'(a)\ge \E_z \hat{\beta}\hat{R}u'(c(\hat{R}(a-\xi)+\hat{Y},\hat{Z}))$, then
\begin{equation*}
u'(a)<\min\set{u'(a),u'(0)}=u'(a),
\end{equation*}
which is a contradiction. Therefore $u'(a)<\E_z \hat{\beta}\hat{R}u'(c(\hat{R}(a-\xi)+\hat{Y},\hat{Z}))$, and \eqref{eq:eulerub} becomes
\begin{align}
u'(\epsilon(z)a)<u'(\xi)&\le \min\set{\E_z \hat{\beta}\hat{R}u'(c(\hat{R}(a-\xi)+\hat{Y},\hat{Z})),u'(0)}\notag \\
&\le \E_z \hat{\beta}\hat{R}u'(c(\hat{R}(a-\xi)+\hat{Y},\hat{Z})).\label{eq:eulerub2}
\end{align}
As in the proof of Proposition~\ref{prop:C1}, the event $\hat{R}=0$ does not affect the expectations in \eqref{eq:suff_cond1} and \eqref{eq:eulerub2}. Therefore without loss of generality we may assume $\hat{R}\ge \delta$ by Assumption~\ref{asmp:R}. Then by $\hat{Y}\ge b$, $a>A$, $\xi<\epsilon(z)a$, and \eqref{eq:Acond}, we have
\begin{align*}
\hat{R}(a-\xi)+\hat{Y}&\ge \delta(1-\epsilon(z))a+\hat{Y}\\
&> \delta\left(1-\max_{z\in \ZZ}\epsilon(z)\right)A+b\ge A.
\end{align*}
Using the fact that $u'$ is strictly decreasing and $c\in \cC_2$, it follows from \eqref{eq:eulerub2} and $\xi<\epsilon(z)a$ that
\begin{align*}
u'(\epsilon(z)a)<u'(\xi)&\le \E_z \hat{\beta}\hat{R}u'(c(\hat{R}(a-\xi)+\hat{Y},\hat{Z}))\\
&\le \E_z \hat{\beta}\hat{R}u'(\epsilon(\hat{Z})(\hat{R}(a-\xi)+\hat{Y}))\\
&\le \E_z \hat{\beta}\hat{R}u'(\epsilon(\hat{Z})\hat{R}[1-\epsilon(z)]a),
\end{align*}
which contradicts \eqref{eq:suff_cond1}. Therefore $T\cC_2\subset \cC_2$.

Now define $c_0(a,z)=a$ and $c_n=T^nc_0$. Since $c_0\in \cC_2$ and $T\cC_2\subset \cC_2$, we have $c_n\in \cC_2$ for all $n$. Therefore $c_n(a,z)\ge \epsilon(z)a$ for all $a>A$ and $z \in \ZZ$. Letting $n\to\infty$, since $c_n\to c$, it follows that $c(a,z)\ge \epsilon(z)a$ for $a>A$. Therefore the limit \eqref{eq:MPC_lim} holds by Theorem \ref{thm:AL_inf}.
\end{proof}

\begin{proof}[Proof of Proposition \ref{prop:aCRRA}]

Let $x^*\in \R_{++}^Z$ be the unique fixed point of $F$ in \eqref{eq:F}, which necessarily satisfies $x^*(z)\ge 1$ for all $z$. Then
\begin{align}
&x^*(z)=\left(1+(K(1-\gamma)x^*)(z)^{1/\gamma}\right)^\gamma \notag \\
\iff & x^*(z)^{1/\gamma}=1+\left(\E_z\hat{\beta}\hat{R}^{1-\gamma}x^*(\hat{Z})\right)^{1/\gamma} \notag \\
\iff & x^*(z)=\E_z\hat{\beta}\hat{R}x^*(\hat{Z})\left(\hat{R}(1-x^*(z)^{-1/\gamma})\right)^{-\gamma}.\label{eq:xstarEq}
\end{align}

Take any $\kappa\in (0,1)$ and define $\epsilon(z)=\kappa x^*(z)^{-1/\gamma}\in (0,1)$. Then
\begin{equation}
\eqref{eq:eulerRatio}=\E_z\hat{\beta}\hat{R}\left(\frac{x^*(\hat{Z})}{x^*(z)}\right)^{\hat{\gamma}/\gamma} \left(\hat{R}[1-\epsilon(z)]\right)^{-\hat{\gamma}}.
\label{eq:eulerRatio2}
\end{equation}
Again we may assume $\hat{R}\ge \delta$ by Assumption~\ref{asmp:R}. Noting that $\hat{\gamma}=\gamma(\hat{c})$, where
\begin{align*}
\hat{c}&=(1-\theta)\epsilon(z)a+\theta \epsilon(\hat{Z})\hat{R}[1-\epsilon(z)]a \\
&\ge (1-\theta)\epsilon(z)a+\theta \epsilon(\hat{Z})\delta [1-\epsilon(z)]a\\
&\ge \min\set{\epsilon(z)a,\epsilon(\hat{Z})\delta [1-\epsilon(z)]a}
\end{align*}
and $\set{\epsilon(z)}_{z\in \ZZ}$ are finitely many fixed numbers in $(0,1)$, it follows that $\hat{c}\to\infty$ uniformly as $a\to\infty$. Since $u$ is aCRRA, we have $\hat{\gamma}\to \gamma$ uniformly as $a\to\infty$. Since $\epsilon(z)=\kappa x^*(z)^{-1/\gamma}$ and $\kappa\in (0,1)$, we can take $A>0$ such that
\begin{align*}
&\E_z\hat{\beta}\hat{R}\left(\frac{x^*(\hat{Z})}{x^*(z)}\right)^{\hat{\gamma}/\gamma}[\hat{R}(1-\epsilon(z))]^{-\hat{\gamma}}\\
&=\E_z\hat{\beta}\hat{R}\left(\frac{x^*(\hat{Z})}{x^*(z)}\right)^{\hat{\gamma}/\gamma}[\hat{R}(1-\kappa x^*(z)^{-1/\gamma})]^{-\hat{\gamma}}\\
&\le  \E_z\hat{\beta}\hat{R}\frac{x^*(\hat{Z})}{x^*(z)}[\hat{R}(1-x^*(z)^{-1/\gamma})]^{-\gamma}=1
\end{align*}
for $a>A$, where the last equation follows from \eqref{eq:xstarEq}. Combined with \eqref{eq:eulerRatio2}, we obtain $\eqref{eq:eulerRatio}\le 1$ for $a>A$. Therefore \eqref{eq:suff_cond1} holds for $a>A$. Finally, \eqref{eq:Acond} holds if $b=\inf Y\ge 0$ is large enough.

If $\delta>1$, then we can take $\kappa\in (0,1)$ such that $\delta(1-\kappa x^*(z)^{-1/\gamma})\ge 1$ for all $z$, so \eqref{eq:Acond} holds for any $b\ge 0$.
\end{proof}

\section{Exponential grid}\label{sec:expgrid}
In many models, the state variable may become negative (\eg, asset holdings), which causes a problem for constructing an exponentially-spaced grid because we cannot take the logarithm of a negative number. Suppose we would like to construct an $N$-point exponential grid on a given interval $(a,b)$. A natural idea to deal with such a case is as follows.

\begin{framed}
\begin{oneshot}[Constructing exponential grid]\quad 
\begin{enumerate}
\item Choose a shift parameter $s>-a$.
\item Construct an $N$-point evenly-spaced grid on $(\log(a+s),\log(b+s))$.
\item Take the exponential.
\item Subtract $s$.
\end{enumerate}
\end{oneshot}
\end{framed}

The remaining question is how to choose the shift parameter $s$. Suppose we would like to specify the median grid point as $c\in (a,b)$. Since the median of the evenly-spaced grid on $(\log(a+s),\log(b+s))$ is $\frac{1}{2}(\log(a+s)+\log(b+s))$, we need to take $s>-a$ such that
\begin{align*}
&c=\exp\left(\frac{1}{2}(\log(a+s)+\log(b+s))\right)-s\\
&\iff c+s=\sqrt{(a+s)(b+s)}\\
&\iff (c+s)^2=(a+s)(b+s)\\
&\iff c^2+2cs+s^2=ab+(a+b)s+s^2\\
&\iff s=\frac{c^2-ab}{a+b-2c}.
\end{align*}
Note that in this case
$$s+a=\frac{c^2-ab}{a+b-2c}+a=\frac{(c-a)^2}{a+b-2c},$$
so $s+a$ is positive if and only if $c<\frac{a+b}{2}$. Therefore, for any $c\in \left(a,\frac{a+b}{2}\right)$, it is possible to construct an exponentially-spaced grid with end points $(a,b)$ and median point $c$.

\end{document}